\documentclass[aps,PRA,reprint,superscriptaddress]{revtex4-2}
\usepackage{graphicx}
\usepackage{braket}
\usepackage{amsmath}
\usepackage{amsthm}
\usepackage{amssymb}
\usepackage[colorlinks,
            linkcolor=blue,
            anchorcolor=black,
            citecolor=blue,
			urlcolor=blue]{hyperref}
			
\usepackage[ruled]{algorithm2e}
\usepackage{makecell}
\usepackage{array}
\usepackage{color}

\newtheorem{proposition}{Proposition}
\newtheorem{definition}{Definition}
\newtheorem{corollary}{Corollary}
\newtheorem{lemma}{Lemma}

\begin{document}

% Use the \preprint command to place your local institutional report
% number in the upper righthand corner of the title page in preprint mode.
% Multiple \preprint commands are allowed.
% Use the 'preprintnumbers' class option to override journal defaults
% to display numbers if necessary
%\preprint{}

%Title of paper
\title{Quantum XYZ cyclic codes for biased noise}

% repeat the \author .. \affiliation  etc. as needed
% \email, \thanks, \homepage, \altaffiliation all apply to the current
% author. Explanatory text should go in the []'s, actual e-mail
% address or url should go in the {}'s for \email and \homepage.
% Please use the appropriate macro foreach each type of information

% \affiliation command applies to all authors since the last
% \affiliation command. The \affiliation command should follow the
% other information
% \affiliation can be followed by \email, \homepage, \thanks as well.
\author{Zhipeng Liang} 
\affiliation{Harbin Institute of Technology, Shenzhen. Shenzhen, 518055, China}
\author{Fusheng Yang}
\affiliation{Harbin Institute of Technology, Shenzhen. Shenzhen, 518055, China}
\author{Zhengzhong Yi}
\email[]{zhengzhongyi@cs.hitsz.edu.cn} 
\affiliation{Hefei National Research Center for Physical Sciences at the Microscale and School of Physical Sciences, University of Science and Technology of China, Hefei 230026, China}
\affiliation{Shanghai Research Center for Quantum Science and CAS Center for Excellence in Quantum Information and Quantum Physics, University of Science and Technology of China, Shanghai 201315, China}
\affiliation{Hefei National Laboratory, University of Science and Technology of China, Hefei 230088, China}
\author{Xuan Wang}
\email[]{wangxuan@cs.hitsz.edu.cn} 
\affiliation{Harbin Institute of Technology, Shenzhen. Shenzhen, 518055, China}
%\email[]{zhengzhongyi@cs.hitsz.edu.cn}
%\homepage[]{Your web page}
%\thanks{}
%\altaffiliation{}

%Collaboration name if desired (requires use of superscriptaddress
%option in \documentclass). \noaffiliation is required (may also be
%used with the \author command).
%\collaboration can be followed by \email, \homepage, \thanks as well.
%\collaboration{}
%\noaffiliation

\date{\today}

\begin{abstract}
In some quantum computing architectures, Pauli noise is highly biased. Tailoring Quantum error-correcting codes to the biased noise may benefit reducing the physical qubit overhead without reducing the logical error rate.  In this paper, we propose a family of quantum XYZ cyclic codes, which are the only one family of quantum cyclic codes with code distance increasing with code length to our best knowledge and have good error-correcting performance against biased noise. Our simulation results show that the quantum XYZ cyclic codes have $50\%$ code-capacity thresholds for all three types of pure Pauli noise and around $13\%$ code-capacity threshold for depolarizing noise. In the finite-bias regime, when the noise is biased towards Pauli $Z$ errors with noise bias ratios $\eta_Z=1000$, the corresponding code-capacity threshold is around $49\%$. Besides, we show that to reach the same code distance, the physical qubit overhead of XYZ cyclic code is much less than that of the XZZX surface code.
\end{abstract}

% insert suggested keywords - APS authors don't need to do this
%\keywords{}

%\maketitle must follow title, authors, abstract, and keywords
\maketitle

% body of paper here - Use proper section commands
% References should be done using the \cite, \ref, and %\label commands
\section{Introduction}
\label{introduction}
Quantum error correcting codes (QECCs)\cite{PhysRevA.52.R2493,PhysRevLett.77.793}often ask for high physical qubit overhead[refs], due to the requirement for correcting all Pauli errors with equal capability. However, in some quantum computing architectures, Pauli noise is highly biased, including superconducting qubits\cite{Aliferis_2009,lescanne2020exponential,chamberland2022building}, quantum dots\cite{shulman2012demonstration} and trap ions\cite{nigg2014quantum}. In these situations, there is no need to ensure the QECCs can correct all Pauli errors with equal capability. Tailoring QECCs to the biased noise, which means making them have strong capability to correct one or two Pauli errors rather than to all of the Pauli errors, may benefit reducing the physical qubit overhead while maintaining low logical error rate.

In recent years, more and more researchers have realized this. For instance, the XZZX surface codes\cite{bonilla2021xzzx} have exhibited remarkably high threshold for all three types of Pauli biased noise. Quantum XZZX cyclic codes proposed in Ref. \cite{xu2023tailored} also exhibit strong error-correcting performance against Pauli $Z$-biased noise. However, if the code distance of XZZX cyclic codes can increase with code length has not been addressed.

So far, to our best knowledge, there is no quantum cyclic code whose code distance has shown or has been proved to have the characteristic of increasing with code length. Our motivation is to design the quantum cyclic codes which have this characteristic and strong error-correcting capability against biased noise.

%In the current quantum computing architectures, qubits still face the problem of being susceptible to noise. Quantum error-correcting code (QECC)\cite{PhysRevA.52.R2493,PhysRevLett.77.793} is the key to solve this problem and to realize universal large-scale fault-tolerant quantum computing\cite{nielsen2001quantum}. When designing QECCs, most researchers usually benchmark their error-correcting performance under depolarizing noise. However, in practice, the noise is highly biased, with dephasing noise being much stronger than all other types of errors\cite{Aliferis_2009,lescanne2020exponential,chamberland2022building}. Thus, designing QECCs with better performance against biased noise is more meaningful in this sense, such as the XZZX surface codes\cite{bonilla2021xzzx}, which have exhibited high threshold when dephasing noise is dominant.

%Quantum XZZX cyclic codes proposed in Ref. \cite{xu2023tailored} also exhibit strong error-correcting performance against Z-biased noise. However, if their code distance can increase with code length has not been addressed. Furthermore, to our best knowledge, so far there is no quantum cyclic code whose code distance can increase with code length and which also has strong performance against biased noise. Our motivation is to design such quantum cyclic codes.

In this paper, we propose a family of quantum XYZ cyclic codes encoding one logical qubit,  whose code distance increases with its code length and stabilizer generators are weight-six XZYYZX-type. First, we prove that by choosing appropriate parameters, the quantum XYZ cyclic codes have repetition code structure under three types of pure Pauli noise. Second, through theoretical analysis and Monte Carlo method proposed in\cite{liang2024determining,bravyi2024high}, we show that the upper bound of code distance of the quantum XYZ cyclic codes increase with code length. Finally, we exploit fully decoupled belief propagation combined with ordered statistics decoding (FDBP-OSD)\cite{yi2023improved} to explore the error-correcting performance of quantum XYZ cyclic codes with the assumption of perfect syndrome measurement. Our simulation results show that the quantum XYZ cyclic codes have $50\%$ code-capacity thresholds for three types of pure Pauli noise and $13\%$ code-capacity thresholds for depolarizing noise. In the finite-bias regime, we consider the noise biased towards Pauli $Z$ errors with noise bias ratio $\eta_Z=1000$, the reason is that in some quantum computing architectures, Pauli $Z$ errors is stronger than all other types of Pauli errors by a factor of nearly $10^3$\cite{Aliferis_2009}. The quantum XYZ cyclic codes show around $49\%$ code-capacity threshold for such noise. Besides, we also compare the physical qubit overhead between the quantum XYZ cyclic codes and the XZZX surface codes. Our results show that to reach the same code distance, the physical qubit overhead of quantum XYZ cyclic codes is much less than that of the XZZX surface code.

This paper is organized as follows. Sect. \ref{Preliminaries} introduces some preliminaries, including ring of circulant, quantum stabilizer codes, quantum cyclic codes, and the Monte Carlo method of determining the upper bound of code distance. In Sect. \ref{Quantum XYZ cyclic codes}, we introduce the quantum XYZ cyclic code construction and study its code dimension and code distance. In Sect. \ref{simulation results}, we explore their error-correcting performance against different noise model. In Sect. \ref{conclusion}, we conclude our work.

\section{Preliminaries}
\label{Preliminaries}
\subsection{Ring of circulants}
\label{Ring of circulants}
For an $l\times l$ circulant matrix $A$ over finite field $\mathbb{F}_{q}$ as follows,
\begin{equation}
A=\left(\begin{array}{cccc}
	a_{0} & a_{1} & \cdots & a_{l-1} \\
	a_{l-1} & a_{0} & \cdots & a_{l-2} \\
	\vdots & \vdots & \ddots & \vdots \\
	a_{1} & a_{2} & \cdots & a_{0}
\end{array}\right)\label{equal1}
\end{equation}
where $a_{0}, a_{1}, \cdots, a_{l-1} \in \mathbb{F}_{q}$, it can be represented as
\begin{equation}
A=a_{0}I+a_{1}J+\cdots a_{l-1}J^{l-1}\label{equal2}
\end{equation}
where $I$ is the $l\times l$ identity matrix and
\begin{equation}
	J={\left(\begin{array}{l l l l l l}{0}&{1}&{0}&{0}&{\cdots}&{0}\\ {0}&{0}&{1}&{0}&{\cdots}&{0}\\{0}&{0}&{0}&{1}&{\cdots}&{0}\\ {\vdots}&{\vdots}&{\vdots}&{\vdots}&{\ddots}&{\vdots}\\ {0}&{0}&{0}&{0}&{\cdots}&{1}\\ {1}&{0}&{0}&{0}&{\cdots}&{0}\end{array}\right)}\label{equal3}
\end{equation}
is the $l\times l$ permutation matrix, which represents the right cyclic shift by one position. It can be verified $J^l=I$, thus the ring of all $l\times l$ circulant matrices over $\mathbb{F}_{q}$ is isomorphic to the ring $\mathbb{F}_{q}^{\langle l\rangle}={\mathrm{\mathbb{F}}}_{q}[x]/(x^{l}+1)$ of polynomials over $\mathbb{F}_{q}$ modulo the polynomial $x^l+1$. Therefore, the circulant matrix $A$ can be uniquely represented by the polynomial $a(x)=a^{0}+\ a^{1}x+\ \cdots\,+a^{l-1}x^{l-1}$.
\subsection{Quantum stabilizer code}
\label{Quantum stabilizer code}
Quantum stabilizer code (QSC)\cite{gottesman1997stabilizer} are the analogue of classical
linear codes in quantum information field. The code space $\mathcal{Q}_{C}$ of an $\left[[n,k,d]\right]$ QSC $C$ is a $2^k$ -dimensional subspace of the Hilbert space ${\mathcal{H}}_{2}^{\otimes n}$, which is stabilized by an Abelian stabilizer group ${\mathcal{S}}\in{\mathcal{G}}_{n}$, where ${\mathcal{G}}_{n}$ is the $n$-fold tensor product of single-qubit Pauli group $\mathcal{G}_1 =\{\pm I,\ \pm iI,\ \pm X,\ \pm iX,\ \pm Y,\ \pm iY,\ \pm Z,\ \pm iZ\}$. More precisely,
\begin{equation}
	\label{equal4}
	\mathcal{Q}_{C}=\{\ket{\varphi}\in{\mathcal{H}}_{2}^{\otimes n}:S\ket{\varphi}= \ket{\varphi},\forall S\in{\mathcal{S}}\}
\end{equation}
The stabilizer group $\mathcal{S}$ can be generated by $k$ independent Pauli operators on $n$ qubits $S_1,\cdots,S_{n-k}\in\mathcal{G}_n$, namely, $\mathcal{S} = \langle S_1, \cdots, S_{n-k} \rangle$. Giving a set of stabilizer generators $S_1,\cdots,S_{n-k}$ of code $C$ is equivalent to explicitly giving the code space $\mathcal{Q}_C$.

Intuitively, the Pauli  operators can be mapped onto finite field $\mathbb{F}_{4}$, namely,
\begin{equation}
\label{equal5}
I\rightarrow0,X\rightarrow1,Z\rightarrow\omega,Y\rightarrow\overline{{{\omega}}}
\end{equation}
Besides, these four operators can also be mapped onto $\mathbb{F}_{2}$, namely,
\begin{equation}
\label{equal6}
I\rightarrow(0,0),X\rightarrow(1,0),Z\rightarrow(0,1),Y\rightarrow(1,1)
\end{equation}
In this way, any operator $E\in\mathcal{G}_n$ acting on $n$ qubits can be represented as a binary vector $\mathbf{e} = (\textbf{e}_x, \textbf{e}_z)$. Based on this $Pauli-to-\mathbb{F}_{2}$ isomorphism, the binary parity-check matrix $H$ of a $\left[\left[n,k\right]\right]$ QSC is a block matrix with dimension $\left(n-k\right)\times2n$, which consists of two $\left(n-k\right)\times n$ binary matrices $H_x$ and $H_z$, namely,
\begin{equation}
\label{equal7}
H=\left(H_x\mid H_z\right)
\end{equation}

The weight of an operator $P\in\mathcal{G}_n$ is defined as the number of qubits on which it acts nontrivially, and we use notation $wt\left(\cdot\right)$ to denote it. For instance, $wt\left(I_1X_2Y_3Z_4\right)=3$.

The logical operators of a QSC $C$ are the set of operators in $\mathcal{G}_n$ which commute with all elements in $\mathcal{S}$ but are not in $\mathcal{S}$. More precisely, the logical operators are the elements of $\mathcal{C}\left(\mathcal{S}\right)/\mathcal{S}$, where $\mathcal{C}\left(\mathcal{S}\right)$ is the centralizer of $\mathcal{S}$ defined as $\mathcal{C}\left(\mathcal{S}\right)=P\in\mathcal{G}_n:SP=PS,\forall S\in\mathcal{S}$. For a $\left[\left[n,k,d\right]\right]$ QSC, we can find $k$ pairs of logical operators $\left({\bar{X}}_j,{\bar{Z}}_j\right)_{j=1,\cdots,k}$ such that ${\bar{X}}_i{\bar{Z}}_j=\left(-1\right)^{\delta_{ij}}{\bar{Z}}_j{\bar{X}}_i$, where $\delta$ is the Kronecker delta, which means for the same pair of logical operators ${\bar{X}}_j,{\bar{Z}}_j$, they are anti-commute, but they commute with other pairs of logical operators. We can see that $\mathcal{C}\left(\mathcal{S}\right)=\left\langle S_1,\cdots,S_{\left(n-k\right)},{\bar{X}}_1,{\bar{Z}}_1,\cdots,{\bar{X}}_k,{\bar{Z}}_k\right\rangle$. The code distance $d$ is defined as the minimum weight of logical operators, namely,
\begin{equation}
\label{equal8}
d = \mathop{\min}_{L \in \mathcal{C(\mathcal{S})} \setminus \mathcal{S}} wt(L)
\end{equation}
\subsection{Quantum cyclic codes}
\label{Quantum cyclic codes}
The stabilizer generators of a quantum cyclic code can be generated cyclically by a given stabilizer generator. For instance, $\left[\left[5,1,3\right]\right]$ quantum code\cite{laflamme1996perfect} is the shortest quantum cyclic code with code distance $d=3$. As shown in Eq. (\ref{equal9}), its stabilizer generators can be obtained by cyclic permutations of $IXZZX$.
\begin{equation}
\label{equal9}
\mathcal{H}={\left(\begin{matrix}I&X&Z&Z&X\\ X&I&X&Z&Z\\Z&X&I&X&Z\\ Z&Z&X&I&X\\ X&Z&Z&X&I\end{matrix}\right)}
\end{equation}
It is easy to verify that the number of independent stabilizer generators in Eq. (\ref{equal9}) is $4$. Besides, $\left[\left[5,1,3\right]\right]$ code is also the shortest code with the effective distance $d_z = 5$ under pure Pauli $Z$ noise, since it exhibits a repetition-code structure for pure Pauli $Z$ noise. Thus, intuitively some quantum cyclic codes might have strong error-correcting performance against certain biased noise. Thus Ref. \cite{xu2023tailored} generalizes the $\left[\left[5,1,3\right]\right]$ code and propose a family of $XZZX$ cyclic codes. However, whether the code distance of this family of $XZZX$ cyclic codes can increase with code length has not been addressed.

For a quantum cyclic code, it can be seen that its parity-check matrix $H_{\mathbb{F}_{4}}$ over $\mathbb{F}_{4}$ is a circulant matrix over $\mathbb{F}_{4}$, while its parity-check matrix $H_{\mathbb{F}_{2}}=\left(H_x\mid H_z\right)$ over $\mathbb{F}_{2}$ is a bipartite circulant matrix, namely, $H_x$ and $H_z$ are both circulant matrices over $\mathbb{F}_{2}$.

In this paper, we will use the parity-check matrix over $\mathbb{F}_{4}$ or $\mathbb{F}_{2}$ alternately depending on circumstances.

\subsection{The Monte Carlo method to determine the upper bound of code distance }
\label{Determining the code distance by Monte Carlo method}
As mentioned in Sect. \ref{Quantum stabilizer code}, the code distance $d$ of a QSC $C$ is defined as the minimum weight of all logical operators. In most cases, for a QSC, we only have its logical operators which are not in the minimum-weight form. To identify the minimum-weight form of its logical operators, one can exploit linear programming method\cite{landahl2011fault} or exhaustive search, but the time complexity of these two methods is exponential\cite{landahl2011fault}. Moreover, it has been proved theoretically that exactly or approximately computing the code distance of QSCs is an NP-complete problem\cite{kapshikar2023hardness,kapshikar2022diagonal}. However, the upper bound of code distance can be determined by more efficient methods, such as Monte Carlo method\cite{rubinstein2016simulation}.

To explicate how to exploit Monte Carlo method to determine the upper bound of code distance, we take a QSC $C$ encoding one logical qubit as an example. The stabilizer group and parity-check matrix of $C$ are denoted by $\mathcal{S}$ and $H$ respectively, and we have its three logical operators $X_L$, $Z_L$ and $Y_L$ which are not in the minimum-weight form. Since any logical operator must commute with all stabilizer generators and anti-commute with another two logical operators, one can solve three decoding problems with a same syndrome $\left(0,\ 0,\ \cdots,0,\ 1,\ 1\right)^T$ and three parity-check matrices ${\left(\begin{matrix}H\\\gamma_Z \\\gamma_Y\end{matrix}\right)}$, ${\left(\begin{matrix}H\\\gamma_X \\\gamma_Y\end{matrix}\right)}$ and ${\left(\begin{matrix}H\\\gamma_X \\\gamma_Z\end{matrix}\right)}$ to obtain three logical operator ${\hat{X}}_L$, ${\hat{Z}}_L$ and ${\hat{Y}}_L$, where $\gamma_X\in X_L\mathcal{S}$, $\gamma_Z\in Z_L\mathcal{S}$ and $\gamma_Y\in Y_L\mathcal{S}$, respectively. If the decoder that is used follows the idea of maximum likelihood decoding, ${\hat{X}}_L$, ${\hat{Z}}_L$ and $ {\hat{Y}}_L$ will be in the minimum-weight form with high probability. Besides, one can set the number trials $T\ \gg\ 1$ and select $\gamma_X^1,\cdots,\gamma_X^T\in X_L\mathcal{S}$,  $\gamma_Z^1,\cdots,\gamma_Z^T\in Z_L\mathcal{S}$ and $\gamma_Y^1,\cdots,\gamma_Y^T\in Y_L\mathcal{S}$ uniformly at random. Then $d_{upper} = \mathop{\min}_{i=1,\ldots,T} \left\{ wt({\hat{X}}_L^i), wt({\hat{Z}}_L^i), wt({\hat{Y}}_L^i) \right\}$ is an upper bound of the code distance which can be systematically improved by increasing the number of trials $T$. 

In this paper, we study the upper bound of the code distance of quantum XYZ cyclic codes encoding one logical qubit through Monte Carlo method proposed in Ref. \cite{liang2024determining,bravyi2024high}, and the decoder we use is FDBP-OSD, which has shown better decoding performance compared with traditional BP-OSD\cite{roffe2020decoding}. Formally, the Monte Carlo method based on FDBP-OSD is given in Algorithm \ref{alg1}.
\begin{algorithm}
	\caption{The Monte Carlo method to determine the upper bound of code distance based on FDBP-OSD}
	\label{alg1}
	\LinesNumbered
	\KwIn{Stabilizer group  $\mathcal{S} = \langle S_1, \cdots, S_{n-1} \rangle$,\\
		three logical operators $X_L,\ Z_L,\ Y_L$ of the $[[N,1]]$ QSC,\\
		$\textbf{s}=\left(0,\ 0,\ \cdots,0,\ 1,\ 1\right)^T$,\\
		the number of trials $T$ and code length $N$.}
	\KwOut{the upper bound of the minimum weight of three logical operators $d_X^{up}$, $d_Z^{up}$ and $d_Y^{up}$.}
	$d_X^{up},\ d_Z^{up},\ d_Y^{up}= N$\\
	\For{$i\gets1\ to\ T$}{
		Randomly selecting $\gamma_X\in X_L\mathcal{S}$, $\gamma_Z\in Z_L\mathcal{S}$ and $\gamma_Y\in Y_L\mathcal{S}$.\\
		Decoding ${\hat{X}}_L=\textbf{FDBP-OSD}\left({\left(\begin{matrix}H\\\gamma_Z \\\gamma_Y\end{matrix}\right)},\ \textbf{s}\right)$, ${\hat{Z}}_L=\textbf{FDBP-OSD}\left({\left(\begin{matrix}H\\\gamma_X \\\gamma_Y\end{matrix}\right)},\ \textbf{s}\right)$, ${\hat{Y}}_L=\textbf{FDBP-OSD}\left({\left(\begin{matrix}H\\\gamma_X \\\gamma_Z\end{matrix}\right)},\ \textbf{s}\right)$.\\
		\If{$wt({\hat{X}}_L) < d_X^{up}$}{
			$d_X^{up} = wt({\hat{X}}_L)$
		}
		\If{$wt({\hat{Z}}_L) < d_Z^{up}$}{
			$d_Z^{up} = wt({\hat{Z}}_L)$
		}
		\If{$wt({\hat{Y}}_L) < d_Y^{up}$}{
			$d_Y^{up} = wt({\hat{Y}}_L)$
		}
	}
	\Return{$d_X^{up}$, $d_Z^{up}$, $d_Y^{up}$}
\end{algorithm}

\section{Quantum XYZ cyclic codes}
\label{Quantum XYZ cyclic codes}
In this section, we first introduce the quantum XYZ cyclic code construction and study its code dimension. Second, we identify the condition under which the quantum XYZ cyclic codes have repetition code structure under pure Pauli noise such that the corresponding codes might have strong error-correcting performance against biased noise. Finally, we investigate the minimum weight of logical operators.
\subsection{Code construction and code dimension}
\label{Code construction and code dimension}
The code construction of quantum XYZ cyclic codes is defined as follows.

\begin{definition}[\textbf{Quantum XYZ cyclic code construction}] 
\label{code construction}
Given two natural numbers $a$ and $b$, the stabilizer generators of the corresponding quantum XYZ cyclic code $C(a,b)$ are obtained by cyclic permutations of $S_1=X\underbrace{I \cdots I}_{b} Z\underbrace{I \cdots I}_{a} YIY\underbrace{I \cdots I}_{a} Z\underbrace{I \cdots I}_{b} X$, namely,
\begin{equation}
\label{stabilizer generator matrix}
\mathcal{H}={\left(\begin{matrix}X\underbrace{I \cdots I}_{b}Z\underbrace{I \cdots I}_{a}YIY\underbrace{I \cdots I}_{a}Z\underbrace{I \cdots I}_{b}X\\ XX\underbrace{I \cdots I}_{b}Z\underbrace{I \cdots I}_{a}YIY\underbrace{I \cdots I}_{a}Z\underbrace{I \cdots I}_{b} \\ \vdots \\ \underbrace{I \cdots I}_{b}Z\underbrace{I \cdots I}_{a}YIY\underbrace{I \cdots I}_{a}Z\underbrace{I\cdots I}_{b}XX \end{matrix}\right)}
\end{equation}
\end{definition}

We refer to $\mathcal{H}$ as the stabilizer generator matrix, and each stabilizer generator is of $XZYYZX$ type, with $b$ identities inserted between $X$ and $Z$, $a$ identities inserted between $Z$ and $Y$, and one identity inserted between two $Y$s.

Notice that for any natural numbers $a$ and $b$, the code length $N=2\left(a+b\right)+7$ is always an odd number, then the following three operators must be logical operators, since they commute with all stabilizers and each one anti-commutes with the other two.
\begin{equation}
	\begin{aligned}
		&X_L=\underbrace{X\cdots X}_{N}\\
		&Z_L = \underbrace{Z \cdots Z}_{N}\\
		&Y_L=\underbrace{Y\cdots Y}_{N}
	\end{aligned}	
\end{equation}

We refer to these three logical operators as $X$-type, $Z$-type and $Y$-type logical operators, respectively, since they only consist of one type of single-qubit Pauli operator.

Next, we prove that the stabilizer group of the corresponding quantum XYZ cyclic code $C(a,b)$ is Abelian.

\begin{proposition}
\label{Commutation}
Given two natural numbers $a$ and $b$, the stabilizer group of the quantum XYZ cyclic code $C(a,b)$ is Abelian.
\end{proposition}
\begin{proof}
To prove the stabilizer group of $C(a,b)$ is Abelian is to prove that any pair of stabilizers commute. According to Definition \ref{code construction}, the parity-check matrix $H$ over $\mathbb{F}_{2}$ of $C(a,b)$ is $H=(H_x\mid H_z)$, where
\begin{equation}
\label{equal11}
H_x={\left(\begin{matrix}1\ \underbrace{0 \cdots 0}_{a+b+1}\ 1\ 0\ 1\ \underbrace{0 \cdots 0}_{a+b+1}\ 1\\
1\ 1\ \underbrace{0 \cdots 0}_{a+b+1}\ 1\ 0\ 1\ \underbrace{0 \cdots 0}_{a+b+1} \\
\vdots \\ \underbrace{0 \cdots 0}_{a+b+1}\ 1\ 0\ 1\ \underbrace{0 \cdots 0}_{a+b+1}\ 1\  1\end{matrix}\right)}
\end{equation}
and
\begin{equation}\label{equal12}
H_z={\left(\begin{matrix}\underbrace{0 \cdots 0}_{b+1}\ 1\ \underbrace{0 \cdots 0}_{a}\ 1\ 0\ 1\ \underbrace{0 \cdots 0}_{a}\ 1\ \underbrace{0 \cdots 0}_{b+1}\\
\underbrace{0 \cdots 0}_{b+2}\ 1\ \underbrace{0 \cdots 0}_{a}\ 1\ 0\ 1\ \underbrace{0 \cdots 0}_{a}\ 1\ \underbrace{0 \cdots 0}_{b}\\ \vdots \\
\underbrace{0 \cdots 0}_{b}\ 1\ \underbrace{0 \cdots 0}_{a}\ 1\ 0\ 1\ \underbrace{0 \cdots 0}_{a}\ 1\ \underbrace{0 \cdots 0}_{b+2}\end{matrix}\right)}
\end{equation}
are both circulant matrix over $\mathbb{F}_{2}$.

As mentioned in Sect. \ref{Ring of circulants}, $H_x$ and $H_z$ can be rewritten as
\begin{equation}
H_x=I+J^{a+b+2}+J^{a+b+4}+J^{2(a+b)+6}
\end{equation}
and
\begin{equation}
H_z=J^{b+1}+J^{a+b+2}+J^{a+b+4}+J^{2a+b+5}
\end{equation}
where $J$ is the $\left[2\left(a+b\right)+7\right]\times[2(a+b)+7]$ permutation matrix. Observe that $J^T=J^{-1}$, we have
\begin{equation}
\label{equal15}
\begin{aligned}
&H_xH_z^T + H_zH_x^T \\
&  = \left(I + J^{a+b+2} + J^{a+b+4} + J^{2(a+b)+6}\right) \\
& \times \left(J^{-(b+1)} + J^{-(a+b+2)} + J^{-(a+b+4)} + J^{-(2a+b+5)}\right) \\
& + \left(J^{b+1} + J^{a+b+2} + J^{a+b+4} + J^{2a+b+5}\right) \\
& \times \left(I + J^{-(a+b+2)} + J^{-(a+b+4)} + J^{-2(a+b)-6}\right) \\
& = \textbf{0}
\end{aligned}
\end{equation}
which means any pair of stabilizers commute and the proof is completed.
\end{proof}

Next, we study the code dimension of quantum XYZ cyclic codes, and show that given any natural numbers $a$ and $b$, the corresponding code $C(a,b)$ encodes either one or three logical qubits.

\begin{proposition}[\textbf{The code dimension of quantum XYZ cyclic codes}]
\label{code dimension}
Given two natural numbers $a$ and $b$, the code dimension $k$ of quantum XYZ cyclic codes $C(a,b)$ is:
\begin{equation}
	\label{dimension}
	\begin{aligned}
	&\text{1. If $b=3l$, $k=1$.}\\
	&\text{2. If $b=3l-1$}, k=\left\{
		\begin{aligned}
			3,\ &if\ (a+1)\ mod\ 3=0 \\
			1,\ &otherwise
		\end{aligned}
		\right.\\
	&\text{3. If $b=3l-2$}, k=\left\{
	\begin{aligned}
		3,\ &if\ a\ mod\ 3=0 \\
		1,\ &otherwise
	\end{aligned}
	\right.\\
	\end{aligned}
\end{equation}
where $l=1,2,\cdots$.
\end{proposition}

\begin{proof}
According to the proof of \textbf{Proposition} 1 in Ref. \cite{panteleev2021degenerate}, the rank of the parity-check matrix $H=(H_x\mid H_z)$ of $C(a,b)$ is $2\left(a+b\right)+7-deg[\gcd(Ax,Bx,x^{2(a+b)+7}+1]$, where $\deg[f(x)]$ is the degree of $f(x)$, and $A\left(x\right)=1+x^{a+b+2}+x^{a+b+4}+x^{2(a+b)+6}$ and $B\left(x\right)=x^{b+1}+x^{a+b+2}+x^{a+b+4}+x^{2a+b+5}$ are the polynomials corresponding to circulant matrices $H_x$ and $H_z$, respectively. Thus, the code dimension $k=deg[\gcd(Ax,Bx,x^{2(a+b)+7}+1]$.
	
Notice that $A\left(x\right)$ and $B\left(x\right)$ can be rewritten as
\begin{equation}
A\left(x\right)=\left(1+x^{a+b+2}\right)\left(1+x^{a+b+4}\right)
\end{equation}
and
\begin{equation}
	B\left(x\right)=x^{b+1}\left(1+x^{a+1}\right)\left(1+x^{a+3}\right)
\end{equation}
thus, we have 
\begin{equation}
\label{equal17}
\begin{aligned}
	&deg[\gcd(Ax,Bx,x^{2(a+b)+7}+1]\\
	&=gcd((a+b+2)(a+b+4),(a+1)(a+3),2a+b+7)
\end{aligned}		
\end{equation}

First, we prove that for any $a$ and $b$, $\gcd{\left(a+b+4,\ 2\left(a+b\right)+7\right)}=1$. Let $m=a+b+4$ and notice that $m=\frac{N+1}{2}$, where $N=2\left(a+b\right)+7$ is the code length, we have $N=2m-1$. Suppose $\gcd{\left(m,\ N\right)}=d$, then $m$ and $m-1=N-m$ must both be multiples of $d$, thus $d$ must be $1$.
	
Second, we prove that if $a+b+2$ is a multiple of $3$, $\gcd{\left(a+b+2,\ 2\left(a+b\right)+7\right)}=3$, or $1$ otherwise. Let $m=a+b+2$ and notice that $m=\frac{N-3}{2}$, where $N=2\left(a+b\right)+7$, we have $N=2m+3$. Suppose $\gcd{\left(m,\ N\right)}=d$, then $m$ and $m+3=N-m$ must both be multiples of $d$, so $\gcd{\left(m,\ m+3\right)=d}$. If $m$ is a multiple of $3$, $d$ must be $3$, or $1$ otherwise.
	
To sum up, if $a+b+2$ is not a multiple of $3$, $gcd\left(\left(a+b+2\right)\left(a+b+4\right),2\left(a+b\right)+7\right)=1$, thus the code dimension $k$ must be $1$. While when $a+b+2$ is a multiple of $3$, to compute the right side of Eq. (\ref{equal17}), we only need to compute $\gcd{\left(a+b+2,(a+1)(a+3)\right)}$ and there are three cases as follows:

\textbf{Case 1:} If $b=3l$, $\gcd\left[a + 3l + 2, (a + 1)(a + 3)\right] \neq 3$.

\textbf{Case 2:} If $b=3l-1$, when $a+1$ is a multiple of 3, $\gcd\left[a+3l+1,\ \ a+1,\ 2\left(a+b\right)+7\right]=3$, or 1 otherwise.

\textbf{Case 3:} If $b=3l-2$, when $a$ is a multiple of 3, $gcd\left[a+3l,\ \ a+3,\ 2\left(a+b\right)+7\right]=3$, or 1 otherwise.

To sum up, we have Eq. (\ref{dimension}), and the proof is completed.
\end{proof}

In the rest of paper, we only consider the quantum XYZ cyclic codes which encode one logical qubit (unless specially noted). There are two reasons. First, if the quantum XYZ cyclic code encodes three logical qubits, it don't have repetition code structure under pure Pauli $Z$ noise, which is proved in Sect. \ref{repetition code structure}. Second, determining the code distance of a QSC, which is the minimum weight of all logical operators, is a NP-complete problem in general\cite{kapshikar2023hardness,kapshikar2022diagonal}, while in this case there are only three logical operators (logical $X$, $Y$ and $Z$ operators) that need to be considered.

\subsection{Quantum XYZ cyclic codes with repetition code structure}
\label{repetition code structure}
In this section, we provide the conditions under which the quantum XYZ cyclic codes have repetition code structure under pure Pauli noise and show that the quantum XYZ cyclic codes encoding three logical qubits don't have repetition code structure under pure Pauli $Z$ noise.
\begin{proposition}[\textbf{The repetition code structure of Quantum XYZ cyclic codes}]
\label{rpcs}
	Given two natural numbers $a$ and $b$, if
	\begin{equation}
		\left\{
		\begin{aligned}
			&\gcd{\left(a+b+2,N\right)}=1\\
			&\gcd{\left(a+b+4,N\right)}=1
		\end{aligned}
		\right.
	\end{equation}
	the corresponding quantum XYZ cyclic code $C(a,b)$ (encoding one logical qubit) has repetition code structure under pure Pauli $Z$ noise. Similarly, if
	\begin{equation}
		\left\{
		\begin{aligned}
			&\gcd{\left(a+1,N\right)}=1\\
			&\gcd{\left(a+3,N\right)}=1
		\end{aligned}
		\right.
	\end{equation}
	or
	\begin{equation}
		\left\{
		\begin{aligned}
			&\gcd{\left(b+1,N\right)}=1\\
			&\gcd{\left(2a+b+5,N\right)}=1
		\end{aligned}
		\right.
	\end{equation}
	$C(a,b)$ has repetition code structure under pure Pauli $X$ or $Y$ noise, where $N=2(a+b)+7$ is the code length.
\end{proposition}
\begin{proof}
	According to Eq. (\ref{stabilizer generator matrix}), the circulant parity-check matrices of $C(a,b)$ under pure Pauli $Z$, $X$ and $Y$ noise are
	\begin{equation}
		H_x=I+J^{a+b+2}+J^{a+b+4}+J^{2(a+b)+6}
	\end{equation}
	\begin{equation}
		H_z=J^{b+1}+J^{a+b+2}+J^{a+b+4}+J^{2a+b+5}
	\end{equation}
	and
	\begin{equation}
		\begin{aligned}
			H_y&=\left(H_x+H_z\right)\ mod\ 2\\
			&=I+J^{b+1}+J^{2a+b+5}+J^{2(a+b)+6}
		\end{aligned}	
	\end{equation}
	respectively. Recall that $J$ is the $N\times N$ permutation matrix and the polynomials corresponding to circulant matrices $H_x$, $H_z$ and $H_y$ are $A\left(x\right)=1+x^{a+b+2}+x^{a+b+4}+x^{2(a+b)+6}=(1+x^{a+b+2})(1+x^{a+b+4})$, $B\left(x\right)=x^{b+1}+x^{a+b+2}+x^{a+b+4}+x^{2a+b+5}=x^b(1+x^{a+1})(1+x^{a+3})$ and $C\left(x\right)=1+x^{b+1}+x^{2a+b+5}+x^{2(a+b)+6}=(1+x^{b+1})(1+x^{2a+b+5})$, respectively.
	
	Under pure Pauli $\sigma (\sigma\in{X,Y,Z})$ noise, proving $C(a,b)$ has repetition code structure is to prove $C(a,b)$ has no logical operators only consisting of single-qubit Pauli $\sigma$ operator whose weight is smaller than $N$.
	
	For $A\left(x\right)=(1+x^{a+b+2})(1+x^{a+b+4})$, when $\gcd{\left(a+b+2,N\right)}=\gcd{\left(a+b+4,N\right)}=1$, $\deg{\left[\gcd{\left(A\left(x\right),x^N+1\right)}\right]}=1$, which means $rank(A(x))=N-1$. In this case, there is only one logical $Z$ operator which consist of $N$ single-qubit Pauli $Z$ operators, thus $C(a,b)$ has repetition code structure under pure Pauli $Z$ noise. The analyses for $B\left(x\right)$ and $C\left(x\right)$ are the same and we do not repeat here.
\end{proof}

\begin{corollary}
	The quantum XYZ cyclic codes encoding three logical qubits don't have repetition code structure under pure Pauli $Z$ noise.
\end{corollary}

\begin{proof}
	According to \textbf{Proposition} \ref{code dimension}, when $b=3l-1$, $(a+1)\ mod\ 3=0$ or $b=3l-2$, $a\ mod\ 3=0$, the corresponding quantum XYZ cyclic code encodes three logical qubits.
	
	\textbf{Case 1:} $b=3l-1$, $(a+1)\ mod\ 3=0$. Suppose $a=3m-1, (m=1,2,3,\cdots)$, the code length of the corresponding code is $N=2(a+b)+7=3(2l+2m+1)$ and $a+b+2=3l+3m$. Thus, $\gcd{\left(a+b+2,N\right)}\neq1$, which don't satisfy the condition in \textbf{Proposition} \ref{rpcs}.
	
	\textbf{Case 2:} $b=3l-2$, $a\ mod\ 3=0$. Suppose $a=3m, (m=1,2,3,\cdots)$, the code length of the corresponding code is $N=2(a+b)+7=3(2l+2m+1)$ and $a+b+2=3l+3m$. Thus, $\gcd{\left(a+b+2,N\right)}\neq1$, which also don't satisfy the condition in \textbf{Proposition} \ref{rpcs}.
\end{proof}

\subsection{The minimum weight of logical $X$ and $Y$ operators}
\label{The minimum weight of logical X and Y operators}

In this section, we prove that the minimum weight of logical $X$ operator can grow linearly with code length, while the minimum weight of logical $Y$ operator is bounded by linear polynomials only related to natural number $b$. For the minimum weight of logical $Z$ operator, we do not give strict proof, but we exploit Algorithm \ref{alg1} to explore it and our results are given in Sect \ref{simulation results}. The results show that the minimum weight of logical $Z$ operator can also grow linearly with code length. These results indicates the code distance of quantum XYZ cyclic codes can increase with code length by increasing $b$.

Proving the minimum weight of logical $\sigma (\sigma\in{X,Y,Z})$ operator of the quantum XYZ cyclic code is equivalent to identifying a subset of stabilizer generators from the stabilizer generator matrix $\mathcal{H}$, namely, Eq. (\ref{stabilizer generator matrix}), such that their product contains the most single-qubit Pauli $\sigma (\sigma\in{X,Y,Z})$ operators. We refer to this subset as a stabilizer selection method.

Next, we prove Lemma \ref{Lemma1}, which will be in the proof of \textbf{Proposition} \ref{The minimum weight of logical X operator}.
\begin{lemma}
\label{Lemma1}
Let $\textbf{y}$ be a binary row vector of size $T$ ($T$ is an odd number), if the code length $N$ of quantum XYZ cyclic code $C$ is an odd multiple of $T$, namely, $N=kT,\ (k=3,5,7,\cdots)$, we have
\begin{equation}
\label{equal18}
\begin{aligned}
&( \underbrace{\textbf{y}, \cdots, \textbf{y}}_{k}) \textbf{H} \\
& = \left( ( \underbrace{\textbf{y}, \cdots, \textbf{y}}_{k} ) \textbf{H}_1,\ ( \underbrace{\textbf{y}, \cdots, \textbf{y}}_{k} ) \textbf{H}_2,\ \cdots,\ (\underbrace{\textbf{y}, \cdots, \textbf{y}}_{k}) \textbf{H}_k \right) \\
& = \left( \textbf{y} \sum_{i=1}^{k} \textbf{H}_{1,i}, \textbf{y} \sum_{i=1}^{k} \textbf{H}_{1,i}, \cdots, \textbf{y} \sum_{i=1}^{k} \textbf{H}_{1,i} \right) \quad 
\end{aligned}
\end{equation}
where $\textbf{H}_{i,j},\ \left(1\le i,j\le k\right)$ are all $T\times T$ matrices, 
\begin{equation}
\label{equal19}
\textbf{H}_{j}=\left(\begin{matrix}\textbf{H}_{1,j}\\ \textbf{H}_{2,j}\\\vdots\\ \textbf{H}_{k,j}\end{matrix}\right),\left(1\le j\le k\right)
\end{equation}
and
\begin{equation}
\label{equal20}
\begin{aligned}
\textbf{H}=\left(\textbf{H}_{1},\ \textbf{H}_{2},\ \cdots,\textbf{H}_{k}\right)=\left(\begin{matrix}\textbf{H}_{1,1}&\textbf{H}_{1,2}&\cdots&\textbf{H}_{1,k}\\ \textbf{H}_{2,1}&\textbf{H}_{2,2}&\cdots&\textbf{H}_{2,k}\\\vdots&\vdots&\ddots&\vdots\\\textbf{H}_{k,1}&\textbf{H}_{k,2}&\cdots&\textbf{H}_{k,k}\end{matrix}\right)
\end{aligned}
\end{equation}

is the parity-check matrix over $\mathbb{F}_{4}$ of $C$.
\end{lemma}

\begin{proof}
Observing that $\textbf{\emph{H}}$ is not only a row cyclic matrix, but also a column cyclic matrix, then we have
\begin{equation}
\label{equal21}
\begin{aligned}
\textbf{\emph{H}}_{j}&=\left(\begin{matrix}\textbf{\emph{H}}_{1,j}\\ \textbf{\emph{H}}_{2,j}\\\vdots\\ \textbf{\emph{H}}_{k,j}\end{matrix}\right)=\left(\begin{matrix}\textbf{\emph{H}}_{1,j}\\
\textbf{\emph{H}}_{1,(j-1)\ mod\ k}\\
\cdots\\
\textbf{\emph{H}}_{1,(j-k+1)\ mod\ k}
\end{matrix}\right)
\end{aligned}
\end{equation}
Thus,
\begin{equation}
\label{equal22}
\begin{aligned}
&(\underbrace{\textbf{\emph{y}}, \cdots, \textbf{\emph{y}}}_{k})\left(\textbf{\emph{H}}_{1},\ \textbf{\emph{H}}_{2},\ \cdots,\ \textbf{\emph{H}}_{k}\right)\\
&=(\underbrace{\textbf{\emph{y}}, \cdots, \textbf{\emph{y}}}_{k}) \left(\begin{matrix}
&\textbf{\emph{H}}_{1,1}&\textbf{\emph{H}}_{1,2}&\cdots&\textbf{\emph{H}}_{1,k}\\
&\textbf{\emph{H}}_{1,k}&\textbf{\emph{H}}_{1,1}&\cdots&\textbf{\emph{H}}_{1,(k-1)}\\
&\vdots&\vdots&\ddots&\vdots\\
&\textbf{\emph{H}}_{1,2}&\textbf{\emph{H}}_{1,k}&\cdots&\textbf{\emph{H}}_{1,1}
\end{matrix}\right)\\
&=\left( \textbf{\emph{y}} \sum_{i=1}^{k} \textbf{\emph{H}}_{1,i}, \textbf{\emph{y}} \sum_{i=1}^{k} \textbf{\emph{H}}_{1,i}, \cdots, \textbf{\emph{y}} \sum_{i=1}^{k} \textbf{\emph{H}}_{1,i} \right)
\end{aligned}
\end{equation}
which completes the proof.
\end{proof}

\textbf{Lemma} \ref{Lemma1} indicates that for a quantum XYZ cycle code of length $N=kT$, if we select stabilizer generators at intervals of $T$ in a specific method (for example, if the 1st stabilizer generator is selected, then the $(1+mT)$th $(1\le m\le k-1)$ stabilizer generators are also selected), the form of the resulting stabilizer is divided into $k$ parts, each of which has the same form.
\begin{proposition}[\textbf{The minimum weight of logical $X$ operator}]
\label{The minimum weight of logical X operator}
Given natural numbers $b$ and $a=2l\left(b+2\right)+l-1,\ \ (l=1,2,\cdots)$, the logical $X$ operator of the minimum-weight form of quantum XYZ cyclic code $C(a,b)$ is
\begin{equation}
\hat{X}_L=\underbrace{\underbrace{\underbrace{I\cdots I}_{b+2}\ X\ \underbrace{I\cdots I}_{b+2}}\cdots\underbrace{\underbrace{I\cdots I}_{b+2}\ X\ \underbrace{I\cdots I}_{b+2}}}_{2l+1}
\end{equation}
and the minimum weight $d_X$ of $\hat{X}_L$ is $2l+1$. 
\end{proposition}

The proof of \textbf{Proposition} \ref{The minimum weight of logical X operator} is given in Appendix \ref{Proof of X}.

Next, we prove that, given any natural number $b$, one can find a stabilizer selection method such that the weight of the resulting logical $Y$ operator $\hat{Y}_L$ converges to a constant which is only related to $b$ with $a$ tending to infinite in Proposition \ref{The upper bound of the minimum weight of the logical Y operator}. This constant is the upper bound of the minimum weight of the logical $Y$ operator.

\begin{proposition}[\textbf{The upper bound of the minimum weight of the logical $Y$ operator}]
\label{The upper bound of the minimum weight of the logical Y operator}
Given two natural numbers $a$ and $b$, when $a$ tends to infinite, the upper bound $d_Y^{up}$ of the minimum weight of the logical $Y$ operator $\hat{Y}_L$ of the quantum XYZ cyclic code $C(a,b)$ is
\begin{equation}
\label{dup_y}
d_Y^{up}=\left\{
\begin{aligned}
&2b+5,\ if\ b=3l\ or\ b=3l-2,\\
&2b+3,\ if\ b=3l-1.
\end{aligned}
\right.
\end{equation}
where $l=1,2,\cdots$.
\end{proposition}

The proof of \textbf{Proposition} \ref{The upper bound of the minimum weight of the logical Y operator} is given in Appendix \ref{Proof of Y}.

\textbf{Proposition} \ref{The minimum weight of logical X operator}, \textbf{Proposition} \ref{The upper bound of the minimum weight of the logical Y operator} and the results in Sect. \ref{simulation results} leads to a conclusion that the code distance of quantum XYZ cyclic codes can increase with code length by increasing $b$.

\section{Simulation results}
\label{simulation results}
In this section, exploiting FDBP-OSD, we study the upper bound of the code distance and error-correcting performance of quantum XYZ cyclic codes through Monte Carlo method. Besides, we compare the physical qubits overhead between quantum XYZ cyclic codes and XZZX surface codes.

\begin{figure}
	\centering
	\includegraphics[width=0.45\textwidth]{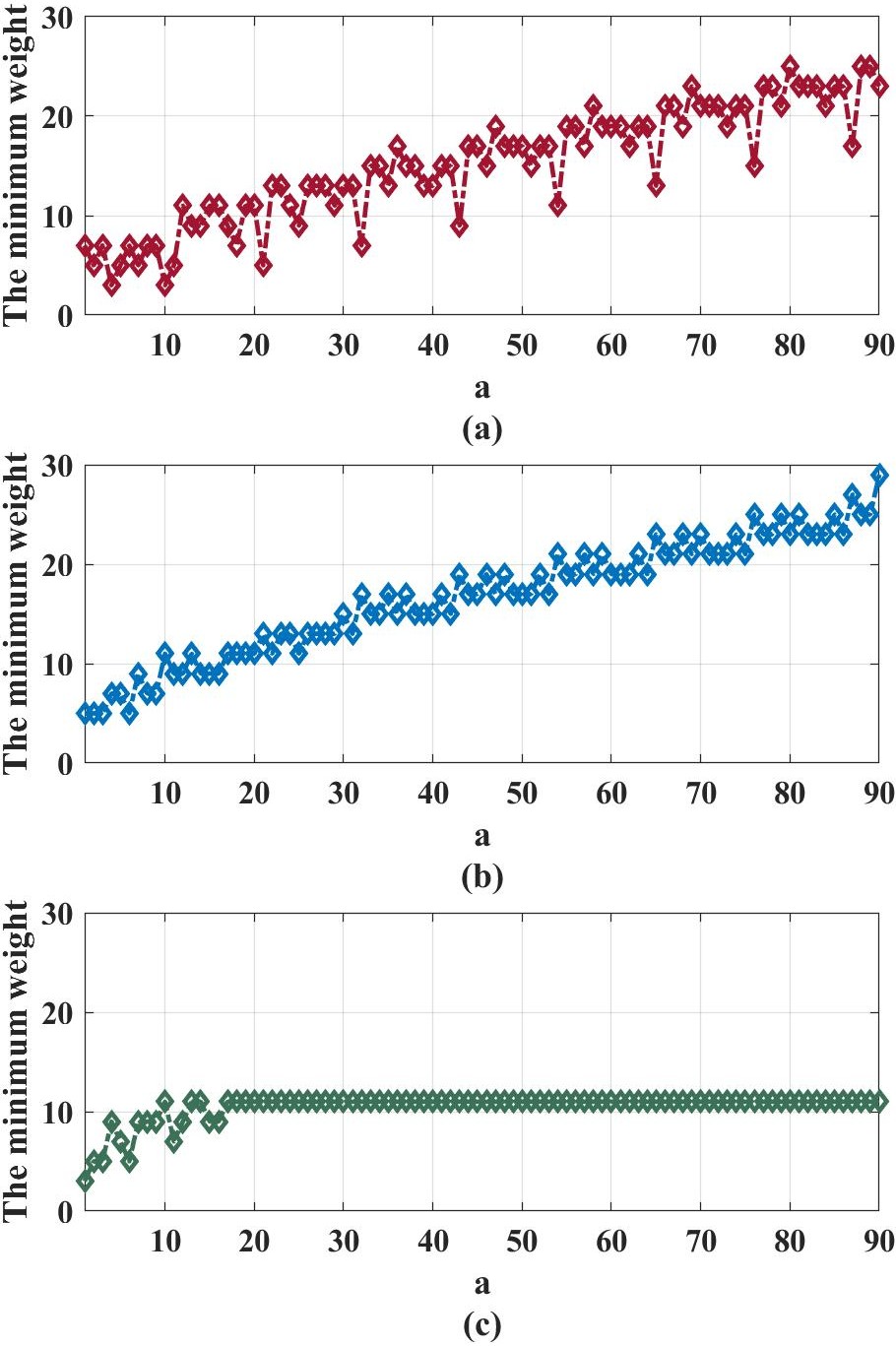}
	\caption{The minimum weight of logical (a) $X$, (b) $Z$, (c) $Y$ operators determined by Monte Carlo method when $b=3$ and $a$ taking value from 1 to 90.}
	\label{minimu_weight}
\end{figure}

\begin{figure}
	\centering
	\includegraphics[width=0.35\textwidth]{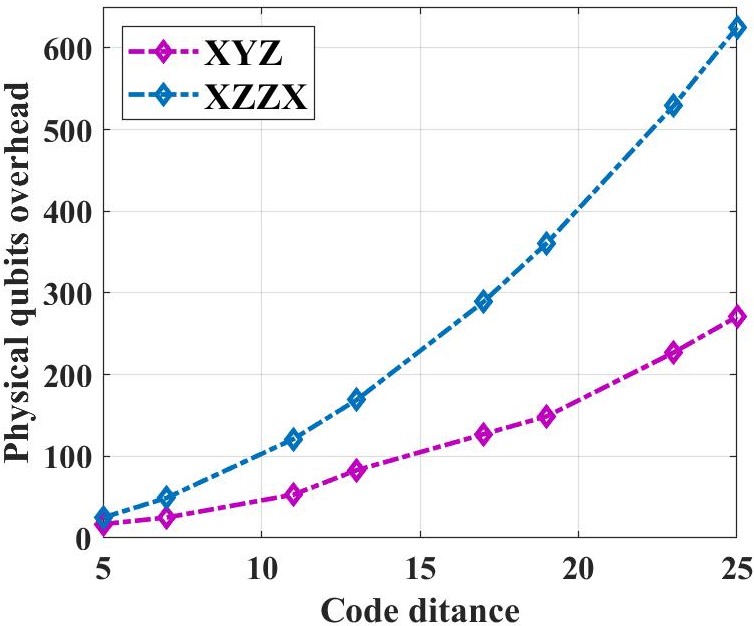}
	\caption{The physical qubits overhead of XYZ cyclic codes and XZZX surface codes with the same code distance.}
	\label{physical_qubits_overhead}
\end{figure}

\begin{figure*}
	\centering
	\includegraphics[width=0.75\textwidth]{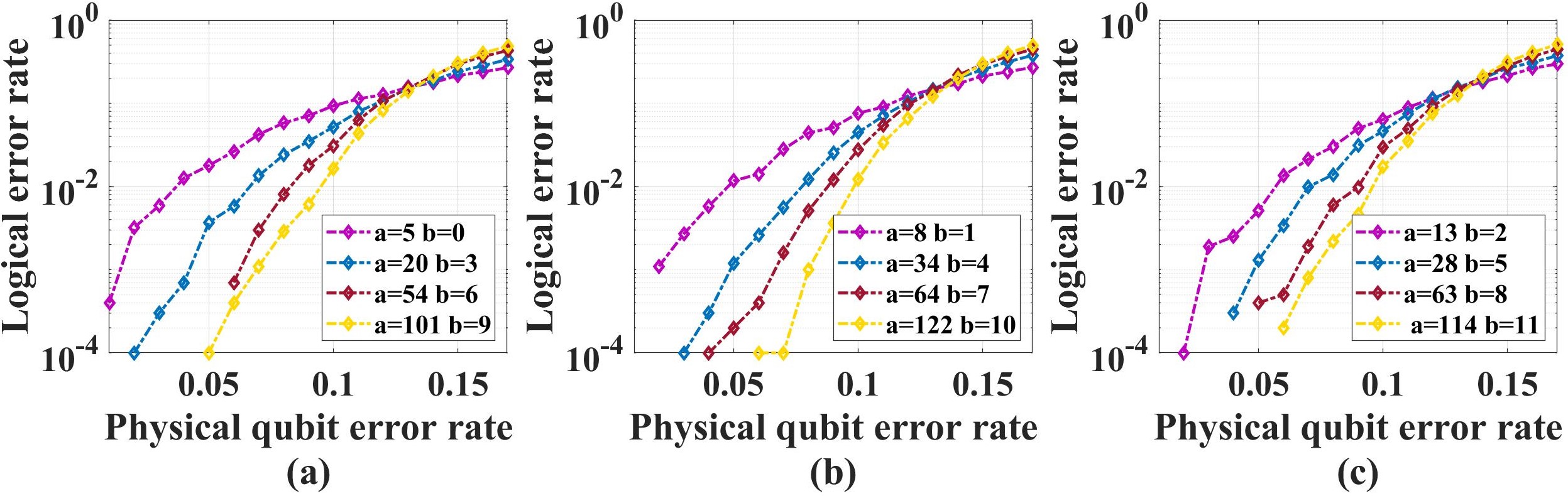}
	\caption{The error-correcting performance of the quantum XYZ cyclic codes against depolarizing noise. (a) $b\ mod\ 3 = 0$, (b) $b\ mod\ 3 = 1$, (c) $b\ mod\ 3 = 2$.}
	\label{dep_code-threshold}
\end{figure*}

\begin{figure*}
	\centering
	\begin{minipage}{0.75\linewidth}
		\centering
		\includegraphics[width=1\linewidth]{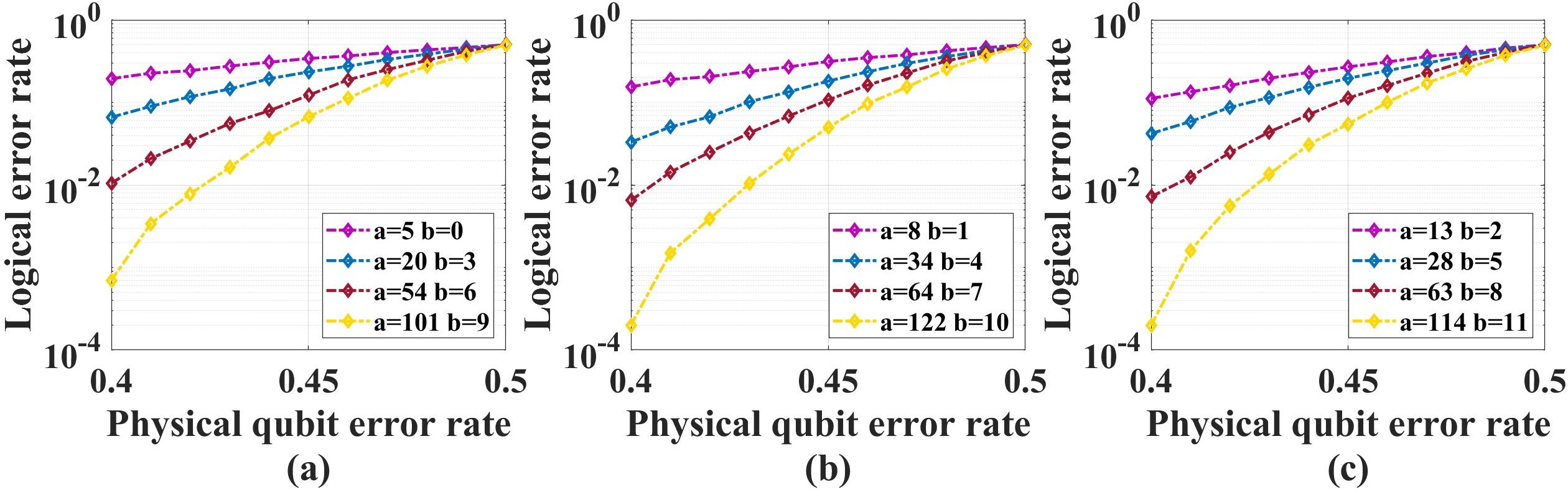}
	\end{minipage}
	%\qquad
	\begin{minipage}{0.75\linewidth}
		\centering
		\includegraphics[width=1\linewidth]{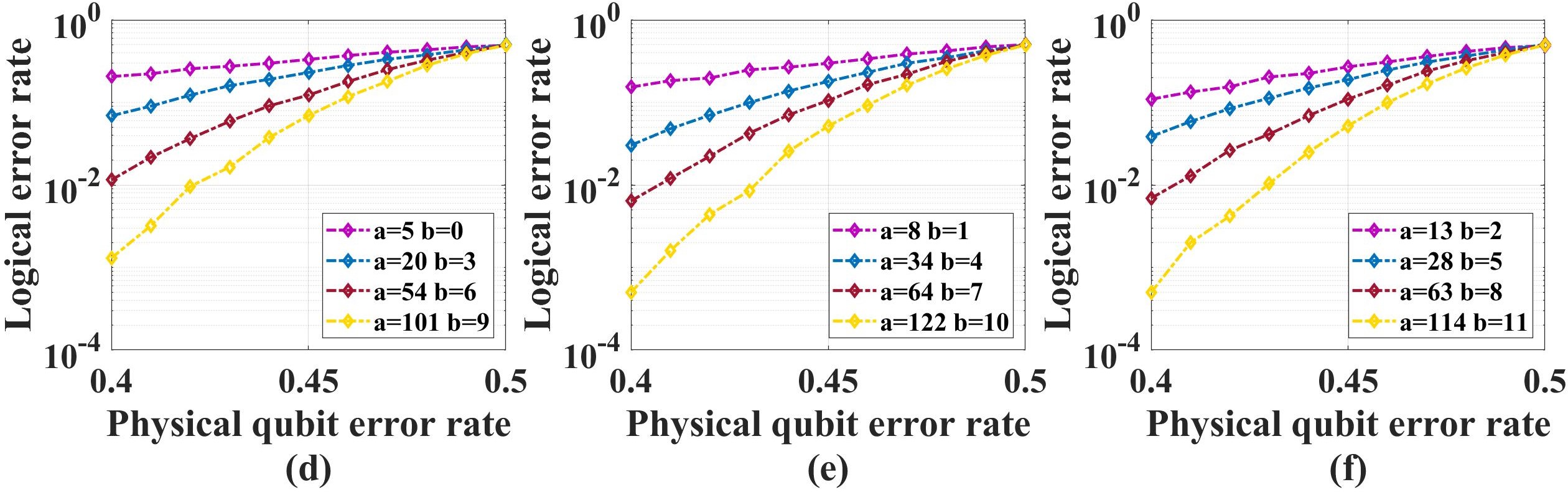}
	\end{minipage}
	%\qquad
	\begin{minipage}{0.75\linewidth}
		\centering
		\includegraphics[width=1\linewidth]{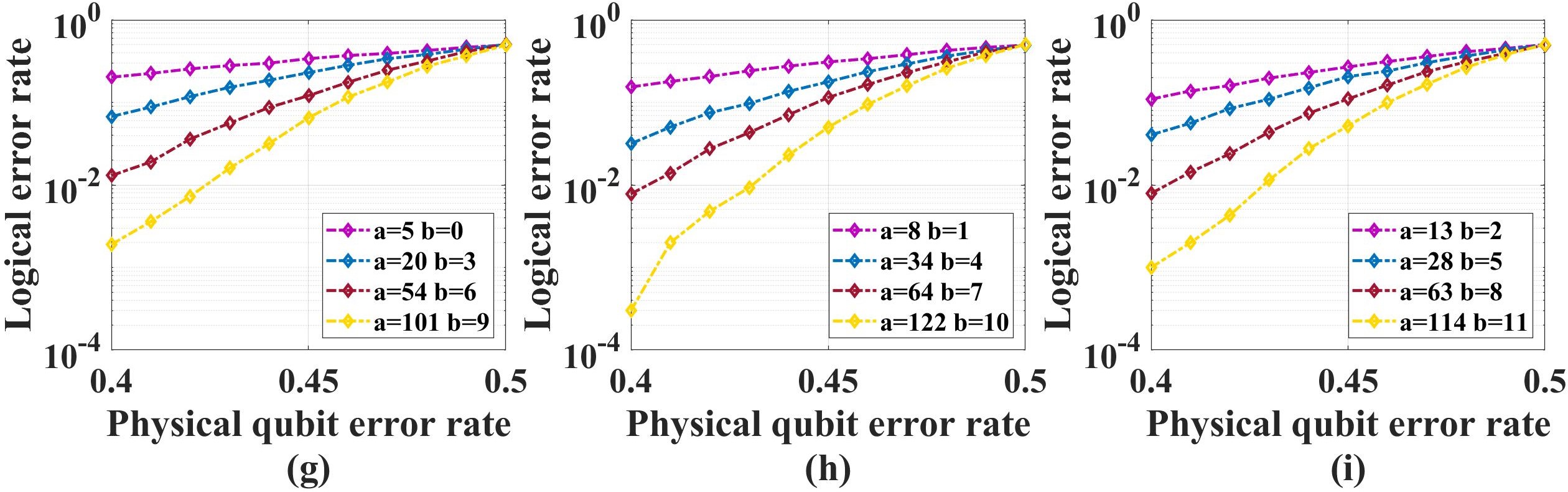}
	\end{minipage}
	\caption{The error-correcting performance of the quantum XYZ cyclic codes against infinitely biased Pauli noise. (a) $b\ mod\ 3 = 0$, infinitely biased Pauli $X$ noise. (b) $b\ mod\ 3 = 1$, infinitely biased Pauli $X$ noise. (c) $b\ mod\ 3 = 2$, infinitely biased Pauli $X$ noise. (d) $b\ mod\ 3 = 0$, infinitely biased Pauli $Z$ noise. (e) $b\ mod\ 3 = 1$, infinitely biased Pauli $Z$ noise. (f) $b\ mod\ 3 = 2$, infinitely biased Pauli $Z$ noise. (g) $b\ mod\ 3 = 0$, infinitely biased Pauli $Y$ noise. (h) $b\ mod\ 3 = 1$, infinitely biased Pauli $Y$ noise. (i) $b\ mod\ 3 = 2$, infinitely biased Pauli $Y$ noise.}
	\label{Pure_code-threshold}
\end{figure*}

\begin{figure*}
	\centering
	\includegraphics[width=0.75  \textwidth]{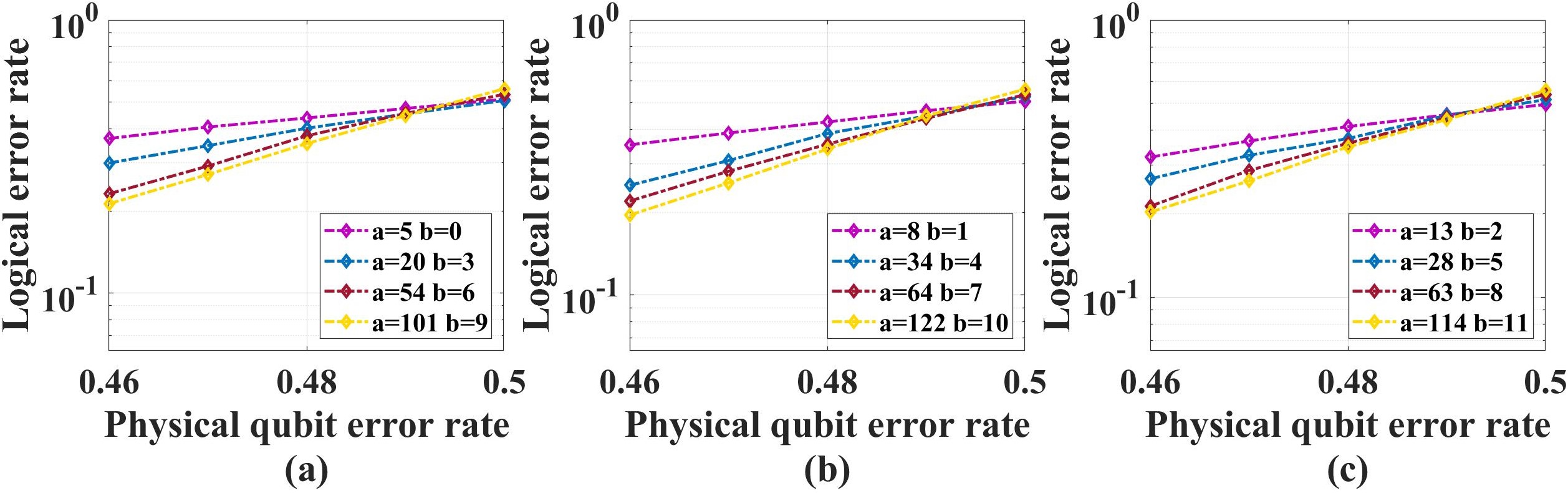}
	\caption{The error-correcting performance of the quantum XYZ cyclic codes against Pauli Z biased noise with biased rate $\eta_Z=1000$. (a) $b\ mod\ 3 = 0$, (b) $b\ mod\ 3 = 1$, (c) $b\ mod\ 3 = 2$.}
	\label{Eta1000_code-threshold}
\end{figure*}

\subsection{Code distance}
\label{Code distance}
We exploit Algorithm \ref{alg1} to study the minimum weight of logical $X$, $Z$ and $Y$ operators for given $a$ and $b$. Our results show that, when $b$ takes a fixed value, the minimum weight of logical $X$ and $Z$ operators trends to increase linearly with the increase of $a$, while that of logical $Y$ operator converge to a constant ($2b+5$ or $2b+3$) with the increase of $a$. Here, we exhibit the results of $b=3$ and a taking value from 1 to 90 in Fig. \ref{minimu_weight} (a)$\sim$(c). The minimum weight of logical $X$ and $Z$ operators trends to increase in Fig. \ref{minimu_weight} (a) and (b) respectively, while that of logical $Y$ operator converge to $2b+5=11$ in Fig. \ref{minimu_weight} (c). Moreover, one can see that, when $a=10,\ 21,\ 32,\ \cdots$, the corresponding minimum weight of logical $X$ operators is $3,\ 5,\ 7,\ \cdots$, which is consistent with \textbf{Proposition} \ref{The minimum weight of logical X operator}.
\subsection{Physical qubits overhead and error-correcting performance}
\label{Error-correcting performance}
When $b$ takes a fixed value, our goal is to identify the minimum value $a$, such that  the corresponding quantum XYZ cyclic code $C(a,b)$ encoding one logical qubit should satisfy the following criteria.

1. $C(a,b)$ should simultaneously has repetition code structure under three types of pure Pauli noise.

2. The code distance of $C(a,b)$ is $2b+5$ ($b\ mod\ 3\neq2$) or $2b+3$ ($b\ mod\ 3=2$).

We refer to the code as the optimal quantum XYZ cyclic code, since it intuitively has good error-correcting performance against both depolarizing noise and biased noise.

For a given fixed $b$, we exploit Algorithm \ref{alg1} to identify the minimum value $a$ and the corresponding minimum weight of three logical operators. Table \ref{The optimal quantum XYZ cyclic codes} shows the results, based on which we compare the physical qubits overhead between the optimal quantum XYZ cyclic codes and XZZX surface codes in Fig. \ref{physical_qubits_overhead}. One can see that, when reaching the same code distance, the physical qubit overhead of quantum XYZ cyclic code is much less than that of the XZZX surface code.
\begin{table}[htbp]
	\begin{center}
		\caption{The optimal quantum XYZ cyclic codes}
		\label{The optimal quantum XYZ cyclic codes}	
		\begin{tabular}{p{1cm}|p{1cm}|c|p{1cm}|p{1cm}|p{1cm}}
			\hline
			$b$ & \ $a$ & \ Code length & $d_X$ & $d_Z$ & $d_Y$\\
			\hline
			0 & 5 & 17 & 5 & 5 & 5 \\
			\hline
			1 & 8 & 25 & 7 & 7 & 7 \\
			\hline
			2 & 13 & 37 & 9 & 9 & 7 \\
			\hline
			3 & 20 & 53 & 11 & 11 & 11 \\
			\hline
			4 & 34 & 83 & 15 & 13 & 13 \\
			\hline
			5 & 28 & 73 & 15 & 13 & 13 \\
			\hline
			6 & 54 & 127 & 19 & 17 & 17 \\
			\hline
			7 & 64 & 149 & 21 & 19 & 19 \\
			\hline
			8 & 63 & 149 & 21 & 19 & 19 \\
			\hline
			9 & 101 & 227 & 29 & 23 & 23 \\
			\hline
			10 & 122 & 271 & 27 & 29 & 25 \\
			\hline
			11 & 114 & 257 & 27 & 27  & 25 \\
			\hline
		\end{tabular}
	\end{center}
\end{table}

According to Eq. (\ref{dup_y}), the optimal quantum XYZ cyclic codes should be divided to three classes based on the value of $b\ mod\ 3$, thus we perform simulation on the three classes of codes separately. 

Fig. \ref{dep_code-threshold} (a)$\sim$(c) show the error-correcting performance of three classes of the optimal quantum XYZ cyclic codes in table \ref{The optimal quantum XYZ cyclic codes} against depolarizing noise. One can observe apparent intersections in these three figures, which are all around $13\%$.

Under infinitely biased noise, simulation results show the code-capacity threshold of three classes of the optimal quantum XYZ cyclic codes in table \ref{The optimal quantum XYZ cyclic codes} all achieve $50\%$ as shown in Fig. \ref{Pure_code-threshold} (a)$\sim$(i), which are consistent with the proof in Sect. \ref{repetition code structure}

In the finite-bias regime, we consider the noise biased towards Pauli $Z$ errors with bias rate $\eta_Z=1000$, the reason is that in some quantum computing architectures, Pauli $Z$ errors is stronger than all other types of Pauli errors by a factor of nearly $10^3$\cite{Aliferis_2009}. Our simulation results show that three classes of the optimal quantum XYZ cyclic codes all achieve around $49\%$ code-capacity threshold, as shown in Fig. \ref{Eta1000_code-threshold} (a)$\sim$(c).

\section {Conclusion}
\label{conclusion}

In this paper, we propose a family of quantum XYZ cyclic codes. Through theoretical analysis and Monte Carlo method based on FDBP-OSD, we show that the code distance of quantum XYZ cyclic codes can increase with code length. To our best knowledge, XYZ cyclic code is the only code whose code distance shows the  characteristic of  increasing with code length. Further, we  exploit Monte Carlo method to identify a set of the optimal quantum XYZ cyclic codes which simultaneusly satisfy the conditions of having repetition code structure under all three types of pure Pauli noise and have the minimum code length for given code distance. We show that the physical qubit overhead of quantum XYZ cyclic code is much less than that of the XZZX surface code when reaching the same code distance.

Our simulation results show that the optimal quantum XYZ cyclic codes have $50\%$ code-capacity thresholds for all three types of pure Pauli noise and around $13\%$ code-capacity threshold for depolarizing noise. In the finite-bias regime, since in some quantum computing architectures, Pauli $Z$ errors is stronger than all other types of Pauli errors by a factor of nearly $10^3$\cite{Aliferis_2009}, we consider the noise biased towards Pauli $Z$ errors with noise bias ratios $\eta_Z=1000$, and the corresponding code-capacity threshold is around $49\%$.

It should be notice that there are still some open questions. The first one is how to strictly prove the minimum weight of logical $Z$ operator increase with code length. The second one is, for given $b$, when $a$ tending to infinite, how to prove there is no logical $Y$ operator whose weight is less than $d_Y^{up}$.
%\newpage
%\section*{End Notes}
\subsection*{Acknowledgements}
This work is supported by the Colleges and Universities Stable Support Project of Shenzhen, China (No.GXWD20220817164856008), Guangdong Provincial Key Laboratory of Novel Security Intelligence Technologies (2022B1212010005), the Colleges and Universities Stable Support Project of Shenzhen, China (No.GXWD20220811170225001) and Harbin Institute of Technology, Shenzhen - SpinQ quantum information Joint Research Center Project (No.HITSZ20230111).

\section*{Data Availability}
All data and Materials are available with the corresponding author on reasonable request.

\appendix

\section{The proof of Proposition \ref{The minimum weight of logical X operator}}
\label{Proof of X}
\begin{proof}
	Notice that $\hat{X}_L$ must anti-commute with $Z$-type and $Y$-type logical operators no matter what value of l is. Thus, to prove $\hat{X}_L$ is the logical $X$ operator of the minimum-weight form, we first prove that $\hat{X}_L$ commutes with all stabilizers (namely, proving $\hat{X}_L$ is a logical $X$ operator), then we prove it is of the minimum-weight form.
	
	The binary representation of $\hat{X}_L$ is
	\begin{equation}
		(L_X^x|L_X^z)=(\underbrace{\underbrace{\underbrace{0\cdots0}_{b+2}\ 1\ \underbrace{0\cdots0}_{b+2}}\cdots\underbrace{\underbrace{0\cdots0}_{b+2}\ 1\ \underbrace{0\cdots0}_{b+2}}}_{2l+1}|\textbf{0}_{1\times N})
	\end{equation}
	and the parity-check matrix over $\mathbb{F}_{2}$ of $C(a,b)$ is $H=(H_x\mid H_z)$, where
	\begin{equation}
		\begin{aligned}
			H_z&=J^{b+1}+J^{(2l+1)\left(b+2\right)+l-1}\\
			&+J^{\left(2l+1\right)\left(b+2\right)+l+1}+J^{\left(4l+1\right)\left(b+2\right)+l+2}
		\end{aligned}
	\end{equation}
	Then we have
	\begin{equation}
		L_X^xH_z^T+L_X^zH_X^T=\ L_X^xH_z^T+\mathbf{0}_{1\times N}H_X^T=\mathbf{0}_{1\times N}
	\end{equation}
	which means $\hat{X}_L$ commutes with all stabilizers.
	
	Proving $\hat{X}_L$ is of the minimum-weight form is equivalent to proving the following stabilizer
	\begin{equation}
		{\hat{S}}^l=\underbrace{\underbrace{\underbrace{X\cdots X}_{b+2}\ I\ \underbrace{X\cdots X}_{b+2}}\cdots\underbrace{\underbrace{X\cdots X}_{b+2}\ I\ \underbrace{X\cdots X}_{b+2}}}_{2l+1}
	\end{equation}
	which is generated by the product of some $XZYYZX$ type stabilizers, contains the most single-qubit Pauli $X$ operators, then the product of $\hat{S}^l$ and $X$-type logical operator $X_L$ is $\hat{X}_L$. Next, we prove it by mathematical induction.
	
	Let $\textbf{\emph{H}}^l$ denote the parity-check matrix over $\mathbb{F}_{4}$ of $C(a,b)$, and $Vec(S)$ denote the row vector representation over $\mathbb{F}_{4}$ of Pauli operator $S$.
	
	1. \textbf{Base case}
	
	When $l=1$ and $a=2\left(b+2\right)$, the code length $N=3[2\left(b+2\right)+1]$, and it can be verified that any operator of weight 1 or 2 must not be logical operator, since these operators must anti-commute with some stabilizers. Thus, when $l=1$, $\underbrace{I\cdots I}_{b+2}\ X\ \underbrace{I\cdots I}_{b+2}\underbrace{I\cdots I}_{b+2}\ X\ \underbrace{I\cdots I}_{b+2}\underbrace{I\cdots I}_{b+2}\ X\ \underbrace{I\cdots I}_{b+2}$ is logical $X$ operator with the minimum-weight form and according to Lemma \ref{Lemma1}, there exists a binary row vector $\textbf{\emph{y}}^1$ of size $2\left(b+2\right)+1$, which corresponds to a method of selecting stabilizer generators at intervals of $2\left(b+2\right)+1$, such that the resulting stabilizer is ${\hat{S}}^1=\underbrace{X\cdots X}_{b+2}\ I\ \underbrace{X\cdots X}_{b+2}\underbrace{X\cdots X}_{b+2}\ I\ \underbrace{X\cdots X}_{b+2}\underbrace{X\cdots X}_{b+2}\ I\ \underbrace{X\cdots X}_{b+2}$, namely,
	\begin{equation}
		\begin{aligned}
			&Vec\left({\hat{S}}^1\right) = \left( \textbf{\emph{y}}^1, \textbf{\emph{y}}^1, \textbf{\emph{y}}^1 \right) \textbf{\emph{H}}^1\\
			& = \left( \textbf{\emph{y}}^1, \textbf{\emph{y}}^1, \textbf{\emph{y}}^1 \right) \left( \textbf{\emph{H}}_1^1, \textbf{\emph{H}}_2^1, \textbf{\emph{H}}_3^1 \right) \\
			&= \left( \textbf{\emph{y}}^1, \textbf{\emph{y}}^1, \textbf{\emph{y}}^1 \right)
			\begin{pmatrix}
				\textbf{\emph{H}}_{1,1}^1 & \textbf{\emph{H}}_{1,2}^1 & \textbf{\emph{H}}_{1,3}^1 \\
				\textbf{\emph{H}}_{1,2}^1 & \textbf{\emph{H}}_{1,1}^1 & \textbf{\emph{H}}_{1,3}^1 \\
				\textbf{\emph{H}}_{1,3}^1 & 
                \textbf{\emph{H}}_{1,2}^1 & \textbf{\emph{H}}_{1,1}^1
			\end{pmatrix} \\
			&= \left( \textbf{\emph{y}}^1 \sum_{k=1}^{3} \textbf{\emph{H}}_{1,k}^1, \textbf{\emph{y}}^1 \sum_{k=1}^{3} \textbf{\emph{H}}_{1,k}^1, \textbf{\emph{y}}^1 \sum_{k=1}^{3} \textbf{\emph{H}}_{1,k}^1 \right) \\
			&= \underbrace{1 \cdots 1}_{b+2} \ 0 \ \underbrace{1 \cdots 1}_{b+2} \ \underbrace{1 \cdots 1}_{b+2} \ 0 \ \underbrace{1 \cdots 1}_{b+2} \ \underbrace{1 \cdots 1}_{b+2} \ 0 \ \underbrace{1 \cdots 1}_{b+2}
		\end{aligned}
	\end{equation}
	which contains the most single-qubit Pauli $X$ operator.
	
	2. \textbf{Inductive Step}
	
	When $l=m$, assume the stabilizer ${\hat{S}}^m=\underbrace{\underbrace{\underbrace{X\cdots X}_{b+2}\ I\ \underbrace{X\cdots X}_{b+2}}\cdots\underbrace{\underbrace{X\cdots X}_{b+2}\ I\ \underbrace{X\cdots X}_{b+2}}}_{2m+1}$  contains the most single-qubit Pauli $X$ operators. According to Lemma \ref{Lemma1}, there exists a binary row vector $\textbf{\emph{y}}^m$ of size $2\left(b+2\right)+1$ such that
	\begin{equation}
		\begin{aligned}
			&Vec\left({\hat{S}}^m\right) = ( \underbrace{\textbf{\emph{y}}^m, \cdots, \textbf{\emph{y}}^m}_{2m+1} ) \textbf{\emph{H}}^m \\
			&=  ( \underbrace{\textbf{\emph{y}}^m, \cdots, \textbf{\emph{y}}^m}_{2m+1} ) \left( \textbf{\emph{H}}_1^m, \textbf{\emph{H}}_2^m,\cdots, \textbf{\emph{H}}_{2m+1}^m \right) \\
			&= ( \underbrace{\textbf{\emph{y}}^m, \cdots, \textbf{\emph{y}}^m}_{2m+1} )
			\begin{pmatrix}
				\textbf{\emph{H}}_{1,1}^m & \textbf{\emph{H}}_{1,2m+1}^m & \cdots &\textbf{\emph{H}}_{1,2}^m \\
				\textbf{\emph{H}}_{1,2}^m & \textbf{\emph{H}}_{1,1}^m & \cdots& \textbf{\emph{H}}_{1,3}^m \\ \vdots&\vdots&\ddots&\vdots \\
				\textbf{\emph{H}}_{1,2m+1}^m & \textbf{\emph{H}}_{1,2m}^m & \cdots& \textbf{\emph{H}}_{1,1}^m
			\end{pmatrix} \\
			&= \left(\underbrace{ \textbf{\emph{y}}^m \sum_{i=1}^{2m+1} \textbf{\emph{H}}_{1,i}^m, \cdots, \textbf{\emph{y}}^m \sum_{i=1}^{2m+1} \textbf{\emph{H}}_{1,i}^m }_{2m+1}\right)\\
			&=\underbrace{ \underbrace{\underbrace{1 \cdots 1}_{b+2}0\underbrace{1 \cdots 1}_{b+2}}\cdots  \underbrace{\underbrace{1 \cdots 1}_{b+2}0 \underbrace{1 \cdots 1}_{b+2}}}_{2m+1}
		\end{aligned}
	\end{equation}
	
	Observe that
	\begin{equation}
		\left\{
		\begin{aligned}
			& \textbf{\emph{H}}_{1,1}^m=\textbf{\emph{H}}_{1,1}^1 \\
			& \textbf{\emph{H}}_{1,2}^m+\textbf{\emph{H}}_{1,m+1}^m=\textbf{\emph{H}}_{1,2}^1 \\
			& \textbf{\emph{H}}_{1,m+2}^m+\textbf{\emph{H}}_{1,2m+1}^m=\textbf{\emph{H}}_{1,3}^1 \\
			& \textbf{\emph{H}}_{1,3}^m=\cdots=\textbf{\emph{H}}_{1,m}^m=\textbf{\emph{H}}_{1,m+3}^m=\cdots=\textbf{\emph{H}}_{1,2m}^m=\mathbf{0}
		\end{aligned}
		\right.
	\end{equation}
	then we have
	\begin{equation}
		\begin{aligned}
			&\textbf{\emph{y}}^m\sum_{i=1}^{2m+1}\textbf{\emph{H}}_{1,i}^m\\
			&=\textbf{\emph{y}}^m\textbf{\emph{H}}_{1,1}^m+\textbf{\emph{y}}^m(\textbf{\emph{H}}_{1,2}^m+\textbf{\emph{H}}_{1,m+1}^m)+\textbf{\emph{y}}^m\sum_{i=3}^{m}\textbf{\emph{H}}_{1,i}^m\\
			&+\textbf{\emph{y}}^m\sum_{i=m+3}^{2m}\textbf{\emph{H}}_{1,i}^m+\textbf{\emph{y}}^m(\textbf{\emph{H}}_{1,m+2}^m+\textbf{\emph{H}}_{1,2m+1}^m)\\
			&=\textbf{\emph{y}}^m(\textbf{\emph{H}}_{1,1}^m+\textbf{\emph{H}}_{1,2}^m+\textbf{\emph{H}}_{1,m+1}^m+\textbf{\emph{H}}_{1,m+2}^m+\textbf{\emph{H}}_{1,2m+1}^m)\\
			&=\ \textbf{\emph{y}}^1\sum_{i=1}^{3}\textbf{\emph{H}}_{1,i}^1
		\end{aligned}
	\end{equation}
	and
	\begin{equation}
		\textbf{\emph{y}}^m=\textbf{\emph{y}}^1
	\end{equation}
	which indicates two methods of selecting stabilizer generators at intervals of $2\left(b+2\right)+1$ corresponding to $\textbf{\emph{y}}^1$ and $\textbf{\emph{y}}^m$ are the same.
	
	When $l=m+1$, we have
	\begin{equation}
		\begin{aligned}
			\textbf{\emph{H}}^{m+1}&=\begin{pmatrix}
				\textbf{\emph{H}}_{1,1}^{m+1} & \textbf{\emph{H}}_{1,2m+3}^{m+1} & \cdots &\textbf{\emph{H}}_{1,2}^{m+1} \\
				\textbf{\emph{H}}_{1,2}^{m+1} & \textbf{\emph{H}}_{1,1}^{m+1} & \cdots& \textbf{\emph{H}}_{1,3}^{m+1} \\ \vdots&\vdots&\ddots&\vdots \\
				\textbf{\emph{H}}_{1,2m+3}^{m+1} & \textbf{\emph{H}}_{1,2m+2}^{m+1} & \cdots& \textbf{\emph{H}}_{1,1}^{m+1}
			\end{pmatrix}\\
			&=\left(\textbf{\emph{H}}_1^{m+1},\ \textbf{\emph{H}}_2^{m+1},\ \cdots,\textbf{\emph{H}}_{2m+3}^{m+1}\right)
		\end{aligned}
	\end{equation}
	Observe that
	\begin{equation}
		\left\{
		\begin{aligned}
			& \textbf{\emph{H}}_{1,1}^{m+1}=\textbf{\emph{H}}_{1,1}^m=\textbf{\emph{H}}_{1,1}^1 \\
			& \textbf{\emph{H}}_{1,2}^{m+1}+\textbf{\emph{H}}_{1,m+2}^{m+1}=\textbf{\emph{H}}_{1,2}^m+\textbf{\emph{H}}_{1,m+1}^m=\textbf{\emph{H}}_{1,2}^1 \\
			& \textbf{\emph{H}}_{1,m+3}^{m+1}+\textbf{\emph{H}}_{1,2m+3}^{m+1}=\textbf{\emph{H}}_{1,m+2}^m+\textbf{\emph{H}}_{1,2m+1}^m=\textbf{\emph{H}}_{1,3}^1 \\
			& \textbf{\emph{H}}_{1,3}^{m+1}=\cdots=\textbf{\emph{H}}_{1,m+1}^{m+1}=\textbf{\emph{H}}_{1,m+4}^{m+1}=\cdots=\textbf{\emph{H}}_{1,2m+2}^{m+1}=\mathbf{0}
		\end{aligned}
		\right.
	\end{equation}
	and transform $\textbf{\emph{H}}^{m+1}$ as
	\begin{equation}
		{\textbf{\emph{H}}^\prime}^{m+1}=({\textbf{\emph{H}}^\prime}^m,\textbf{\emph{H}}_{m+1}^{m+1},\textbf{\emph{H}}_{2m+2}^{m+1})
	\end{equation}
	where 
	\begin{equation}
		\begin{aligned}
			{\textbf{\emph{H}}^\prime}^m&=(\textbf{\emph{H}}_1^{m+1},\ \textbf{\emph{H}}_2^{m+1},\ \cdots,\textbf{\emph{H}}_m^{m+1},\textbf{\emph{H}}_{m+2}^{m+1},\textbf{\emph{H}}_{m+3}^{m+1},\\
			&\cdots,\textbf{\emph{H}}_{2m+1}^{m+1},\textbf{\emph{H}}_{2m+3}^{m+1})
		\end{aligned}
	\end{equation}
	Since $\textbf{\emph{H}}_{1,m+1}^{m+1}=\textbf{\emph{H}}_{1,2m+2}^{m+1}=\textbf{0}$, we have
	\begin{equation}
		\begin{aligned}
			&\textbf{\emph{y}}^m\sum_{i=1}^{2m+1}\textbf{\emph{H}}_{1,i}^m
			=\textbf{\emph{y}}^m\left(\textbf{\emph{H}}_{1,m+1}^{m+1}+\textbf{\emph{H}}_{1,2m+2}^{m+1}\right)+\textbf{\emph{y}}^m\sum_{i=1}^{2m+1}\textbf{\emph{H}}_{1,i}^m\\
			&=\textbf{\emph{y}}^m\sum_{i=1}^{2m+3}\textbf{\emph{H}}_{1,i}^{m+1}\\
			&=(\underbrace{\textbf{\emph{y}}^m,\cdots,\textbf{\emph{y}}^m}_{2m+3})\textbf{\emph{H}}_1^{m+1}\\
			&=\underbrace{1\cdots1}_{b+2}\ 0\ \underbrace{1\cdots1}_{b+2}
		\end{aligned}
	\end{equation}
	which indicates that the row vector generated by selecting some rows of ${\textbf{\emph{H}}}^{\prime m}$ have $2(2m+1)(b+2)\ \ 1$s at most and the corresponding selecting method is $(\underbrace{\textbf{\emph{y}}^m,\cdots,\textbf{\emph{y}}^m}_{2m+3})$. Thus, in this case, for the row vector generated by selecting some rows of matrix $\left(\textbf{\emph{H}}_{m+1}^{m+1},\textbf{\emph{H}}_{2m+2}^{m+1}\right)$, we have
	\begin{equation}
		\begin{aligned}
			(\underbrace{\textbf{\emph{y}}^m,\cdots,\textbf{\emph{y}}^m}_{2m+3})\left(\textbf{\emph{H}}_{m+1}^{m+1},\textbf{\emph{H}}_{2m+2}^{m+1}\right)\\
			=\underbrace{1\cdots1}_{b+2}\ 0\ \underbrace{1\cdots1}_{b+2}\underbrace{1\cdots1}_{b+2}\ 0\ \underbrace{1\cdots1}_{b+2}
		\end{aligned}
	\end{equation}
	which has $2(b+2)$ 1s at most. Thus, the row vector generated by selecting some rows of $\textbf{\emph{H}}^{m+1}$ have $2(2m+3)(b+2)\ \ 1$s at most, which means the corresponding stabilizer $\hat{S}^{m+1}=\underbrace{\underbrace{\underbrace{X\cdots X}_{b+2}\ I\ \underbrace{X\cdots X}_{b+2}}\cdots\underbrace{\underbrace{X\cdots X}_{b+2}\ I\ \underbrace{X\cdots X}_{b+2}}}_{2m+3}$ contains the most single-qubit Pauli $X$ operators. Thus, the product of ${\hat{S}}^{m+1}$ and $X$-type logical operator $X_L$ is $\underbrace{\underbrace{\underbrace{I\cdots I}_{b+2}\ X\ \underbrace{I\cdots I}_{b+2}}\cdots\underbrace{\underbrace{I\cdots I}_{b+2}\ X\ \underbrace{I\cdots I}_{b+2}}}_{2m+3}$, which is of minimum weight $2(m+1)+1$, and the proof is completed.
\end{proof}

\section{The proof of Proposition \ref{The upper bound of the minimum weight of the logical Y operator}}
\label{Proof of Y}
\begin{proof}
Here, we define a bitwise multiplication operator $\star$ of two matrices $A$ and $B$, and $A\star B$ represents that each row of $A$ and $B$ is performed bitwise multiplication and this operator is used in the following proof.

As mentioned in Sect. \ref{The minimum weight of logical X and Y operators}, we only consider the quantum XYZ cyclic codes encoding one logical qubit, thus there is one redundant stabilizer generator in $\mathcal{H}$. To facilitate the analysis, we remove the first stabilizer from $\mathcal{H}$ and obtain the following stabilizer generator matrix $\mathcal{H}^\ast$:
	\begin{equation}\label{H*}
            \begin{aligned}
		\mathcal{H}^\ast=\begin{pmatrix}
			XX\underbrace{I\cdots I}_b Z\underbrace{I\cdots I}_a Y IY\underbrace{I\cdots I}_a Z\underbrace{I\cdots I}_b\\
			IXX\underbrace{I\cdots I}_b Z\underbrace{I\cdots I}_a YIY\underbrace{I\cdots I}_a Z\underbrace{I\cdots I}_{b-1}\\
			\vdots\\
			\underbrace{I\cdots I}_b Z\underbrace{I\cdots I}_a YIY\underbrace{I\cdots I}_a Z\underbrace{I\cdots I}_b XX
		\end{pmatrix}
         \end{aligned}
	\end{equation}
	which contains $2\left(a+b\right)+6$ independent stabilizer generators. Notice that the stabilizers on the $i$th and $[2\left(a+b\right)+7-i]$th rows of Eq. (\ref{H*}) are reverses of each other.
	
	Next, we prove that when $b=3l$, selecting one stabilizer generator from $\mathcal{H}$ at intervals of 1, the weight of the resulting logical $Y$ operator $\hat{Y}_L$ is $2b+5$. While selecting three stabilizer generators from $\mathcal{H}$ intervals of 3, the weight of $\hat{Y}_L$ is $2b+5$ when $b=3l-2$ or $2b+3$ when $b=3l-1$.
	
	\textbf{Case 1:} $b=3l$.
	
	Selecting stabilizer generators from $\mathcal{H}^\ast$ at intervals of 1, one can obtain the following residual stabilizer generator matrix $\mathcal{H}^\prime$:
	\begin{equation}
		\mathcal{H}^\prime=\begin{pmatrix}
			XX\underbrace{I\cdots I}_b Z\underbrace{I\cdots I}_a YIY\underbrace{I\cdots I}_a Z\underbrace{I\cdots I}_b\\
			IIXX\underbrace{I\cdots I}_b Z\underbrace{I\cdots I}_a YIY\underbrace{I\cdots I}_a Z\underbrace{I\cdots I}_{b-2}\\
			\vdots\\
			\underbrace{I\cdots I}_{b-1}Z\underbrace{I\cdots I}_a YIY\underbrace{I\cdots I}_a Z\underbrace{I\cdots I}_b XXI
		\end{pmatrix}
	\end{equation}
	Observe that $\mathcal{H}^\prime$ can be rewritten as $\mathcal{H}^\prime=\mathcal{H}_1\star\mathcal{H}_2\star\mathcal{H}_3$, where $\mathcal{H}_1=\begin{pmatrix}
		XX\underbrace{I\ \cdots\ I}_{2a+2b+5}\\
		IIXX\underbrace{I\cdots I}_{2a+2b+3}\\
		\vdots\\
		\underbrace{I\cdots I}_{2a+2b+4}XXI
	\end{pmatrix}$, $\mathcal{H}_2 =\begin{pmatrix}
		\underbrace{I\cdots I}_{b+2}Z\underbrace{I\cdots I}_{2a+3}Z\underbrace{I\cdots I}_ b\\
		\underbrace{I\cdots I}_{b+4}Z\underbrace{I\cdots I}_{2a+3}Z\underbrace{I\cdots I}_{b-2}\\
		\vdots\\
		\underbrace{I\cdots I}_{b-1}Z\underbrace{I\cdots I}_{2a+3}Z\underbrace{I\cdots I}_{b+3}
	\end{pmatrix}$, and $\mathcal{H}_3=\begin{pmatrix}
		\underbrace{I\cdots I}_{a+b+3}YIY\underbrace{I\cdots I}_{a+b+1}\\
		\underbrace{I\cdots I}_{a+b+5}YIY\underbrace{I\cdots I}_{a+b-1}\\
		\vdots\\
		\underbrace{I\cdots I}_{a+b}YIY\underbrace{I\cdots I}_{a+b+4}
	\end{pmatrix}$. The product of all elements in $\mathcal{H}_1$, $\mathcal{H}_2$ and $\mathcal{H}_3$ are $P_1=\underbrace{X\cdots X}_{2a+2b+6}I$, $P_2=\underbrace{I\cdots I}_{b+1}\underbrace{Z\cdots Z}_{2a+4}\underbrace{I\cdots I}_{b+2}$, and $P_3=\underbrace{I\cdots I}_{a+b+2}YIIY\underbrace{I\cdots I}_{a+b+1}$, respectively. Thus, the resulting stabilizer $\hat{S}$, which is the product of $P_1$, $P_2$ and $P_3$, is
	\begin{equation}
		\hat{S}=P_1P_2P_3=\underbrace{X\cdots X}_{b+1}\underbrace{Y\cdots Y}_{a+1}IYYI\underbrace{Y\cdots Y}_{a-1}\underbrace{X\cdots X}_{b+1}I
	\end{equation}
	Since the $Y$-type logical operator is $Y_L=\underbrace{Y\cdots Y}_N$, the product of $\hat{S}$ and $Y_L$ is
	\begin{equation}
		\hat{Y}_L=\underbrace{Z\cdots Z}_{b+1}\underbrace{I\cdots I}_{a+1}YIIY\underbrace{I\cdots I}_{a-1}\underbrace{Z\cdots Z}_{b+1}Y
	\end{equation}
	whose weight is $2b+5$ and is only related to $b$.
	
	\textbf{Case 2:} $b=3l-1$ and $b=3l-2$.
	
	Since the proofs for $b=3l-1$ and $b=3l-2$ are similar, we only give the proof of $b=3l-1$ here. When $b=3l-1$, according to Proposition \ref{code dimension}, $a\ mod\ 3\neq 2$. Let $k=\left\lfloor\frac{a}{3}\right\rfloor$ (where $\left\lfloor\cdot\right\rfloor$ represents the floor operation) and assume $a \ mod\ 3=0$ (when $a \ mod\ 3=1$, the proof is the same), selecting three stabilizers from $\mathcal{H}^\ast$ at intervals of 3, the residual stabilizer generator matrix $\mathcal{H}^\prime$ is:
	\begin{equation}
		\mathcal{H}^\prime=\begin{pmatrix}
			XX\underbrace{I\cdots I}_b Z\underbrace{I\cdots I}_a YIY\underbrace{I\cdots I}_ a Z\underbrace{I\cdots I}_b\\
			IXX\underbrace{I\cdots I}_ b Z\underbrace{I\cdots I}_a YIY\underbrace{I\cdots I}_a Z\underbrace{I\cdots I}_{b-1}\\
			IIXX\underbrace{I\cdots I}_ b Z\underbrace{I\cdots I}_ a YIY\underbrace{I\cdots I}_ a Z\underbrace{I\cdots I}_{b-2}\\
			\vdots\\
			\underbrace{I\cdots I}_{b-2}Z\underbrace{I\cdots I}_a YIY\underbrace{I\cdots I}_a Z\underbrace{I\cdots I}_b XXII\\
			\underbrace{I\cdots I}_{b-1}Z\underbrace{I\cdots I}_a YIY\underbrace{I\cdots I}_a Z\underbrace{I\cdots I}_b XXI\\
			\underbrace{I\cdots I}_ b Z\underbrace{I\cdots I}_ a YIY\underbrace{I\cdots I}_ a Z\underbrace{I\cdots I}_b XX
		\end{pmatrix}
	\end{equation}
	Rewriting $\mathcal{H}^\prime$ as $\mathcal{H}^\prime=\mathcal{H}_1\star\mathcal{H}_2\star\mathcal{H}_3$, where
		$\mathcal{H}_1=\begin{pmatrix}
			XX\underbrace{I\cdots I}_{2a+2b+5} \\
			IXX\underbrace{I\cdots I}_{2a+2b+4} \\
			IIXX\underbrace{I\cdots I}_{2a+2b+3} \\
			\vdots\\
			\underbrace{I\cdots I}_{2a+2b+3} XXII\\
			  \underbrace{I\cdots I}_{2a+2b+4} XXI\\
			\underbrace{I\cdots I}_{2a+2b+5} XX
		\end{pmatrix}$,
            $\mathcal{H}_2=\begin{pmatrix}
			\underbrace{I\cdots I}_{b+2} Z \underbrace{I\cdots I}_{2a+3} Z \underbrace{I\cdots I}_ b \\
			\underbrace{I\cdots I}_{b+3} Z \underbrace{I\cdots I}_{2a+3} Z \underbrace{I\cdots I}_ {b-1}\\
			\underbrace{I\cdots I}_{b+4} Z \underbrace{I\cdots I}_{2a+3} Z \underbrace{I\cdots I}_ {b-2}\\
			\vdots\\
			\underbrace{I\cdots I}_{b-2} Z \underbrace{I\cdots I}_{2a+3} Z \underbrace{I\cdots I}_ {b+4}\\
			\underbrace{I\cdots I}_{b-1} Z \underbrace{I\cdots I}_{2a+3} Z \underbrace{I\cdots I}_ {b+3}\\
			\underbrace{I\cdots I}_{b} Z \underbrace{I\cdots I}_{2a+3} Z \underbrace{I\cdots I}_ {b+2}
		\end{pmatrix}$,and
            $\mathcal{H}_3=\begin{pmatrix}
			\underbrace{I\cdots I}_{a+b+3}YIY \underbrace{I\cdots I}_{a+b+1}\\
			\underbrace{I\cdots I}_{a+b+4}YIY \underbrace{I\cdots I}_{a+b}\\
			\underbrace{I\cdots I}_{a+b+5}YIY \underbrace{I\cdots I}_{a+b-1}\\
			\vdots\\
			\underbrace{I\cdots I}_{a+b-1}YIY \underbrace{I\cdots I}_{a+b+5}\\
			\underbrace{I\cdots I}_{a+b}YIY \underbrace{I\cdots I}_{a+b+4}\\
			\underbrace{I\cdots I}_{a+b+1}YIY \underbrace{I\cdots I}_{a+b+3}
		\end{pmatrix}$
	respectively. The product of all elements in $\mathcal{H}_1$, $\mathcal{H}_2$ and $\mathcal{H}_3$ are $P_1=\underbrace{\underbrace{XII}\cdots\underbrace{XII}}_{k+l}XIIIX\underbrace{\underbrace{IIX}\cdots\underbrace{IIX}}_{k+l}$, $P_2=\underbrace{Z\cdots Z}_{3l}\underbrace{\underbrace{ZII}\cdots\underbrace{ZII}}_{k-l+1}\underbrace{I\cdots I}_{6l-1} \underbrace{\underbrace{IIZ}\cdots\underbrace{IIZ}}_{k-l+1}\underbrace{Z\cdots Z}_{3l}$ and $P_3=\underbrace{\underbrace{IYY}\cdots\underbrace{IYY}}_{k+l} IYIYI\underbrace{\underbrace{YYI}\cdots\underbrace{YYI}}_{k+l}$, respectively. Thus, the resulting stabilizer $\hat{S}$, which is the product of $P_1$, $P_2$ and $P_3$, is
	\begin{equation}
		\begin{aligned}
		\hat{S}=&P_1P_2P_3\\
			=&\underbrace{\underbrace{YXX}\cdots\underbrace{YXX}}_{l}\underbrace{Y\cdots Y}_{3(k-l+1)}\underbrace{\underbrace{XYY}\cdots\underbrace{XYY}}_{l-1} XYIYX\\&\underbrace{\underbrace{YYX}\cdots\underbrace{YYX}}_{l-1}\underbrace{Y\cdots Y}_{3(k-l+1)}\underbrace{\underbrace{XXY}\cdots\underbrace{XXY}}_{l}
		\end{aligned}	
	\end{equation}
	Then the product of $\hat{S}$ and $Y$-type logical operator $Y_L$ is
	\begin{equation}
		\begin{aligned}
        \hat{Y}_L=&\underbrace{\underbrace{IZZ}\cdots\underbrace{IZZ}}_{l}\underbrace{I\cdots I}_{3(k-l+1)}\underbrace{\underbrace{ZII}\cdots\underbrace{ZII}}_{l-1}ZIYIZ\\&\underbrace{\underbrace{IIZ}\cdots\underbrace{IIZ}}_{l-1} \underbrace{I\cdots I}_{3(k-l+1)}\underbrace{\underbrace{ZZI}\cdots\underbrace{ZZI}}_{l}
            \end{aligned}
	\end{equation}
	whose weight is $6l+1$, which is equal to $2b+3$.
\end{proof}
\newpage
\bibliography{reference}
\end{document}